\documentclass[12pt,a4paper]{article}

\usepackage[T1]{fontenc}

\usepackage[top=1in, bottom=1in, left=1in, right=1in]{geometry}

\usepackage{setspace}
\usepackage[latin9]{inputenc}

\usepackage[dvipsnames]{xcolor}

\usepackage{amsfonts, amssymb, amsmath, amsthm}

\usepackage{hyperref}
\definecolor{darkblue}{rgb}{0, 0, 0.5}
\hypersetup{
     colorlinks = true,
     citecolor = blue,
     urlcolor = blue,
     linkcolor = blue
}
\usepackage[all]{hypcap}
\usepackage[nameinlink, capitalize, noabbrev]{cleveref}

\usepackage{changes}
\definechangesauthor[name=Pasha, color=BrickRed]{P}

\usepackage{enumerate}
\usepackage{enumitem} %
\usepackage{cases}
\usepackage{bm}

\usepackage[authoryear]{natbib}
\usepackage{graphicx}
\usepackage{float}

\usepackage{tikz}
\usetikzlibrary{decorations}
\usepackage{rotating}
\usetikzlibrary{arrows}

\makeatletter

\newtheorem{assu}{Assumption}
\newtheorem{stand}{Standing Assumption}

\crefname{assu}{Assumption}{Assumptions}
\crefname{stand}{Assumption}{Assumptions}
\crefname{prop}{Proposition}{Propositions}
\crefname{thm}{Theorem}{Theorems}
\crefname{lemm}{Lemma}{Lemmas}

\newtheorem{coro}{Corollary}
\newtheorem{lemm}{Lemma}
\newtheorem{thm}{Theorem}

\newtheorem{prop}{Proposition}
\newtheorem*{thm*}{Theorem}

\newcommand\eop{\hfill $\square$}

\renewcommand{\geqslant}{\geq}
\usepackage{bm}

\title{Scoring and Favoritism \\ in Optimal Procurement Design}
\author{
\setcounter{footnote}{1}
    Pasha Andreyanov\thanks{
	HSE University, Moscow.
Email: \href{pandreyanov@gmail.com}{pandreyanov@gmail.com}
	}
	\and Ilia Krasikov\thanks{
		HSE University, Moscow.
		Email: \href{krasikovis.main@gmail.com}{krasikovis.main@gmail.com}
	}
	\and Alex Suzdaltsev\thanks{
		HSE University, Saint Petersburg
		Email: \href{asuzdaltsev@gmail.com}{asuzdaltsev@gmail.com}
	}
}
\begin{document}

\maketitle

\begin{abstract}
We study buyer-optimal procurement mechanisms when quality is contractible. When some costs are borne by every participant of a procurement auction regardless of winning, the classic analysis should be amended. We show that an optimal \emph{symmetric} mechanism is a scoring auction with a score function that may be either flatter or steeper than classically. This depends on the relative degrees of information asymmetry over the all-pay and winner-pay costs.

However, the symmetry of the optimal mechanism is not granted due to the presence of all-pay costs. When ex-post efficiency is less important than the duplication of costs, favoritism becomes optimal. We show that, depending on the degree of convexity of costs, the solution takes one of two novel formats with a partially asymmetric treatment of firms, which we call a score floor and a score ceiling auction. Interestingly, these auctions feature side payments from or to the buyer, which has nothing to do with corruption.

\end{abstract}

\medskip

\noindent \textbf{JEL Classfication:} D44, D82

\medskip


\newpage

\onehalfspacing

\newpage
\section{Introduction}

In recent decades, procurement agencies across the world have come under increasing pressure to improve performance and deliver projects faster. The reason is that the traditional approach, when the contract is awarded to the lowest price bidder, fails to capture the trade-off between costs and quality of procurement. At the same time, quality may represent a large portion of the buyer's utility. As a result, numerous alternative auction designs have emerged, see \cite{molenaar2007alternative} for an overview.

One such design is the scoring auction, where the contract is awarded to the firm with the best combination of price and quality. It was shown to be superior to the traditional approach theoretically (see \cite{che1993design,asker2008properties,asker2010procurement}) when quality is contractible, and the associated costs are winner-pay. A special case when quality is represented by the speed of delivery is known as ``A+B auctions'' in road construction (see \cite{lewis2011procurement}). On the other hand, when quality is not contractible, the average bid auction (see \cite{albano2006bid, decarolis2014awarding}, \cite{decarolis2018comparing}) and the low-ball lottery auction (see \cite{lopomo2022optimal}) have been proposed.

Our goal is to study optimal mechanisms when a firm's quality is endogenous and contractible, but some of the associated costs are \emph{all-pay}. That is, some quality-dependent costs are paid not by the winner only, but by all participants. The firm's experience and past performance, among others, have this property. This makes quality special in the sense that the firm has it chosen before the main (bidding) stage of the mechanism and fully commits to it. Given this, it is natural to condition allocation and payments directly on qualities (along with the standard ``messages'') in any mechanism. We put extra effort into defining such mechanisms formally  \cref{sec:setup}. Also, we assume that there is no communication among firms at the quality-choice stage. This implies that in any equilibrium, the quality of every firm can depend on its efficiency type only, but not on the efficiency types of competitors. This is consistent with an empirically relevant setting in which a firm competes for several similar contracts but against different sets of rivals but has the same quality in all of those competitions.

We begin with the analysis of optimal \emph{symmetric} mechanisms in \cref{sect:symmetric}. We show that, under mild conditions, an optimal symmetric mechanism is a scoring auction with a quasi-linear scoring rule, see \cref{optimalsym,symmimplement}. In the optimal mechanism, quality is distorted downwards to account for the informational rents, see
\cref{downward}, which was also the case in the standard model with winner-pay costs only \citep{che1993design}. However, it is not independent of the number of bidders, as in \cite{che1993design}, but is decreasing instead, see \cref{symmcompare}. The intuition is that, as a firm's market share shrinks, it can not afford to put the same level of investment upfront. The optimal scoring rule also depends on the number of bidders~--- which is also nonstandard~--- but the comparative statics is more complicated and depends on the relative elasticities of winner-pay and all-pay costs, see \cref{elmatters}.

In section~\cref{sec:suboptimal}, we proceed with two arguments showing that even with ex-ante symmetric firms, symmetric mechanisms can be sub-optimal. One such argument is simple: if there is little or no private information, sole-sourcing (i.e., procurement from a single supplier) is optimal, see \cref{smallalpharesult}. The second argument is subtler:  regardless of the amount of private information, the optimal symmetric mechanism may be dominated by another mechanism where some of the worst types are treated asymmetrically. This happens when the elasticity of the investment costs with respect to quality is high enough,  see \cref{symmetricbad2}.

The most intriguing (to our minds) part of the paper is the analysis of optimal mechanisms without the restriction to symmetry in \cref{sect:asym}. For two firms, we identify four different mechanisms that could be optimal.\footnote{For a narrow class of environments, these four mechanisms comprise a complete classification of optimal mechanisms, see \cref{fig:Ex1} in \cref{sec:example}.} The first two are the traditional scoring auction and sole-sourcing. The former maximizes ex-post efficiency, while the latter minimizes all-pay costs. The other two mechanisms, born from the tension between these forces, are novel and present a surprising combination of symmetric and asymmetric design, integrating aspects of both scoring and favoritism. The key factors determining the shape of the optimal mechanism are the informational asymmetry and the curvature of the marginal all-pay costs, see \cref{main result}.

Our first novel mechanism can be thought of as an optimal symmetric allocation with the type of the favored firm censored from above (right); thus, we dub it \textit{right-censored}. It can be implemented by a scoring auction where both firms face a \emph{score floor} (a reserve score) $\underline{S}$, but the favored bidder gets a monetary \textit{bonus} for exceeding $\underline{S}$, see \cref{implement-1}. The score floor and the bonus work in concert to ensure that the favored bidder's worst (rightmost) types participate in the mechanism, while the unfavored bidder's worst types do not.

Our second novel mechanism can be thought of as an optimal symmetric allocation with the type of the unfavored firm censored from below (left); thus, we dub it \textit{left-censored}. It can be implemented by another modification of a scoring auction. Instead of a score floor, one should impose a \emph{score ceiling} $\bar{S}$ so that the bids with a score higher than $\bar{S}$ do not count. As in the case of score floor, one should also supplement these rules with a side payment but  \emph{from} the favored bidder rather than to her, see \cref{implement-2} and, thus, we call this transfer a \emph{kickback}.

The proof of the auction-style implementation serves two purposes. First, it gives hope that these mechanisms could be used in practice. Second, it shows that all the incentive constraints are satisfied; see \cref{th1}, \cref{th2}, and \cref{th3}. This is far from obvious since firms can deviate in both quality and price and unlike in \cite{che1993design} where there are no all-pay costs, these two decisions cannot be decoupled. We explain the nuances of these double deviations in \cref{sec:symmimplement,sec:implementation}.

We provide a number of additional results in \cref{sec:additional}. What determines the optimal degree of symmetry? We show that when the importance of private information grows, the optimal mechanism becomes more symmetric, see \cref{compstat1}. The intuition is that selecting the ex-post efficient firm becomes more important than avoiding the duplication of investment costs. This also leads to a somewhat surprising conclusion that the efficient mechanism could be \emph{less} symmetric than the optimal one, see \cref{effvsopt}, as informational rents may amplify the role of private information.

Finally, we provide some insight into the optimality with more than two firms. In \cref{1orall}, we focus on a natural family of ``restricted entry'' mechanisms in which one allows only $k$ bidders to enter and then lets them play the optimal symmetric mechanism. One may think that choosing some intermediate $k$ (say, $n/2$) may be a good option since it balances the desire to preserve ex-post efficiency and avoid the duplication of investment costs. Under a mild non-parametric assumption, we show that, in fact, the optimal $k$ in this case is always either 1 or $n$. Thus, a partial entry restriction is never optimal; if a buyer wishes to create asymmetry by pure entry restriction, she should restrict entry all the way to monopoly.

\subsection{Related literature}\label{sec:literature}
Our paper contributes to the vast literature on scoring auctions, see \cite{che1993design,branco1997design,asker2008properties,asker2010procurement,nishimura2015optimal} for important theoretical contributions, and \cite{adani2018procuring,lewis2011procurement,decarolis2016past} for empirical. In these auctions, quality of works is used as part of the selection criterion but, with the exception of \cite{decarolis2016past}, the associated costs are winner-pay. \footnote{\cite{decarolis2016past} document long-lasting blackouts associated with traditional price-only procurement auctions for electricity works in Italy. As a result, a scoring auction was used with past performance playing the role of quality. Since past performance was measured in past contracts, all associated costs are effectively all-pay from the current auction's viewpoint.}

Notable contributions were made in \cite{manelli1995optimal,lopomo2022optimal} where quality is exogenous and not contractible, and in \citep{tan1992entry,piccione1996cost,arozamena2004} where quality  is chosen before learning one's type. 

Another important strand of the literature is where investment or entry decisions are made after learning one's type, see \cite{celik2009optimal,zhang2017auctions,gershkov2021theory}. Contrary to our paper, in their settings (i) an agent's action is \emph{not contractible} which precludes the use of scores which are a focus of the present paper, and (ii) an agent's action does not directly benefit the principal. 

Assumptions (i) and (ii) can be motivated, for instance, by the Spectrum auctions, as well as some procurement auctions where the buyer is not concerned with the quality of works and so the firm invests only into costs reduction. But, in many applications, especially in public procurement, quality is important for the buyer. In these settings, quality is traditionally used in the selection criterion, either as part of the scoring formula or as minimal quality requirements, and thus is observable.\footnote{There is a recent dynamic of moving away from the price-only auctions towards the most economically advantageous tenders, i.e., scoring auctions and alike, see  European Union Directive 2014/24/EU.}

A salient manifestation of the difference between our setting and that of \cite{celik2009optimal,zhang2017auctions,gershkov2021theory} is that in the mentioned papers, there is essentially only one mechanism in which the most efficient agent always gets the good/contract. Once this allocation rule is fixed, the agents' actions (quality) are hidden and thus uniquely defined. In contrast, in our setting, multiple action profiles are compatible with the same allocation rule, and the designer can choose (and enforce) the best ones. Despite these differences, we were able to leverage some of their technical results.\footnote{The symmetric results in our paper hold in greater generality than the asymmetric ones, which require a separable environment. This assumption can not be relaxed without losing the link with \cite{zhang2017auctions}. We discuss the fine details of this link in \cref{zhangzhang}.}

Finally, there exists a sizeable literature on corruption in scoring auctions, e.g. \cite{celentani2002corruption}, \cite{burguet2004competitive}, \cite{burguet2017procurement}, \cite{huang2019procurement}, \cite{huang2019empirical}. However, the favoritism in our paper is not due to corruption, but by design.

\subsection{Organization of the paper}
In \cref{sec:setup} we set the environment. In \cref{sect:symmetric} we derive the optimal symmetric mechanism and study its comparative statics. In \cref{sec:suboptimal} we show the suboptimality of symmetric mechanisms and derive the optimal asymmetric mechanism in \cref{sect:asym}. In \cref{sec:additional}, we consider various factors that determine asymmetry of the optimal mechanism, as well as the ex-ante exclusion of firms, and we conclude in \cref{sec:conclude}.

The technical proofs for \cref{sect:symmetric,sec:suboptimal,sect:asym,sec:additional} are contained in \cref{app:symmetric,app:suboptimal,app:asym,app:additional}.

\section{Setup}
\label{sec:setup}

Consider a single buyer (principal) who wishes to procure a contract for which there are $n$ potential firms (agents). The quality of the work to be procured is endogenous. Upon privately learning her cost parameter (type) $\theta_i \in [0,1]$, the firm chooses quality $q_i \geqslant 0$, which is perfectly observed by the buyer and is contractible.\footnote{$\theta = 1$ is the worst, most costly type.}  Following \cite{che1993design}, we assume that quality is one-dimensional. The contract can be allocated among the agents in shares $z_i\geqslant 0$. Also following \cite{che1993design}, we posit that the good must be procured in any case, that is, $\sum_{i=1}^n z_i=1$.\footnote{The buyer always procures the good, if her outside option is sufficiently low.}

Quality is costly to produce, and each supplier will incur per-unit production costs $C^{P}(q_i,\theta_i)$ if he is selected for the contract. The novel feature is that in addition to the production (winner-pay) costs, each supplier $i$ also incurs investment costs $C^{I}(q_i,\theta_i)$, which are sunk before the auction. Thus, the investment costs are \emph{all-pay} costs that are borne regardless of winning the contract. Types  $\theta_i$ are independently distributed with cdf $F(\theta_i)$ and strictly positive  density $f(\theta_i)$. 

Note that the cost functions are not firm-specific, so all firms are ex-ante symmetric. Due to the all-pay nature of $C^I(q_i,\theta_i)$, a firm's (agent's) payoff is given by 
\begin{equation*}
\Pi(z_i,t_i,q_i,\theta_i)=t_i-C^P(q_i,\theta_i)z_i-C^I(q_i,\theta_i),
\end{equation*}
where $t_i$ is the monetary transfer to firm $i$.

The buyer's (principal's) payoff is given by $$U(\mathbf{z},\mathbf{t},\mathbf{q})=\sum_{i=1}^n (V(q_i)z_i-t_i)$$ for some function $V(q_i)$ {representing her preference towards the quality of works}.\footnote{We will use boldface to denote a vector indexed by $i = 1, \ldots, n$, throughout the paper.}

We impose the following standard assumptions on the primitives of the model.
\begin{stand}\label{cost_properties} \noindent
\begin{enumerate}
\item $C^P_{q}>0$, $C^P_{\theta}>0$, $C^I_{q}>0$, $C^I_{\theta}>0$,
\item $C^P_{q\theta}>0$, $C^I_{q\theta}>0$,
\item $C^I(0,\theta)=0$, $C^I_q(0,\theta)=0$,
\item $V'(q) \geqslant 0$, $V''(q) \leq 0$.
\end{enumerate}
\end{stand}

Because quality is observable and contractible, we ought to include $q_i$ directly as determinants of allocation and transfers. 
Thus, our definition of a mechanism is somewhat non-standard. A \textbf{mechanism} $\mathcal{M}$ is a  collection of message sets $M_i$ and functions $z_i(\mathbf{m},\mathbf{q})$,  $t_i(\mathbf{m},\mathbf{q})$, for all  $i=1,\ldots,n$. 

We consider the following timing of the game:
\begin{center}
    \begin{tikzpicture}
    \label{timeline}
        \draw[->, line width=.5pt, >=stealth] (-8,0) -- (4,0);
    \foreach \x in {-7,-5,-3, -1,1,3}
    \draw (\x cm,3pt) -- (\x cm,-3pt);
    
    \draw (-5.5,0) node[above=3pt,anchor=south] {\begin{turn}{30}
            buyer announces $\mathcal{M}$
        \end{turn}};
    \draw (-3.3,0) node[above=3pt,anchor=south] {\begin{turn}{30} 
        firm learns her type $\theta_i$
        \end{turn}};
    \draw (-1,0) node[above=3pt,anchor=south] {\begin{turn}{30}  firm chooses her quality $q_i$ \end{turn}};
    \draw (1.2,0) node[above=3pt,anchor=south] {\begin{turn}{30} firm submits her message $m_i$ \end{turn}};
    \draw (3.1,0) node[above=3pt,anchor=south] {\begin{turn}{30} buyer assigns outcomes $z_i, t_i$\end{turn}};
    \draw (4,0) node[above=3pt,anchor=south] {\begin{turn}{30} payoffs realize \end{turn}};

    \draw (-7,0) node[below=2pt] {$ 1 $} ;
    \draw (-5,0) node[below=2pt] {$ 2 $};
    \draw (-3,0) node[below=2pt] {$ 3 $} ;
    \draw (-1,0) node[below=2pt] {$ 4 $};
    \draw (1,0) node[below=2pt] {$ 5 $};
    \draw (3,0) node[below=2pt] {$ 6 $};
    \end{tikzpicture}
\end{center}

Importantly, given this timing, quality $q_i$ may depend only on \emph{own} type $\theta_i$ while allocation $z_i$ and transfers $t_i$ may depend on the types of \emph{all} firms through the messages sent. That is, we assume that the decisions about investment in quality are \emph{independent}~--- a bidder cannot condition her quality on her competitor's types. This is a common assumption in the literature on auctions with endogenous valuations (\cite{gershkov2021theory}, \cite{celik2009optimal}).  

We posit that the firms will play a (principal-suggested) Bayes-Nash equilibrium (BNE) of any given mechanism. We will use the following shorthands: $$\sigma_i(\theta_i) := (m_i(\theta_i),q_i(\theta_i)), \quad \sigma_{-i}(\theta_{-i}) := (m_{-i}(\theta_{-i}),q_{-i}(\theta_{-i})), \quad \bm\sigma(\bm\theta) := (\sigma_i(\theta_i),\sigma_{-i}(\theta_{-i})).$$ %
The buyer's problem can be written as
\begin{align}
 \hypertarget{P}{\mbox{(P) }}& \max\limits_{\mathcal{M},\bm\sigma} \mathbb{E}_{\bm\theta}U(\bm z(\bm\sigma(\bm\theta)),\bm t(\bm\sigma(\bm\theta)),q_1(\theta_1), \ldots, q_n(\theta_n)) \nonumber \\
\mbox{s.t. } & \mathbb{E}_{\theta_{-i}}\Pi_i(z_i(\bm\sigma(\bm\theta)),t_i(\bm\sigma(\bm\theta)),q_i(\theta_i),\theta_i)\geq \nonumber \\
& \mathbb{E}_{\theta_{-i}}\Pi_i(z_i(m_i',q''_i,\sigma_{-i}(\theta_{-i})),t_i(m_i',q''_i,\sigma_{-i}(\theta_{-i})),q''_i,\theta_i),\label{IC} \\
& \mathbb{E}_{\theta_{-i}}\Pi_i(z_i(\bm\sigma(\bm\theta)),t_i(\bm\sigma(\bm\theta)),q_i(\theta_i),\theta_i)\geq 0, \label{PC}\\
& z_i(\bm\sigma(\bm\theta))\geqslant 0, \ \sum_{i=1}^n z_i(\bm\sigma(\bm\theta))=1
\end{align}
for all $i$, $(m'_i, q''_i) \in M_i \times \mathbb{R}_{+}$, and $\theta_i, \bm\theta$ in the support.
We stress that a firm can make a \emph{double deviation} $(m_i',q''_i)$, which is an important part of our analysis below.

\subsection{Problem relaxations}

 We shall proceed by making several relaxations of the incentive constraints. 
 
First, we focus on the case where by sending a message $m'_i$, a firm mimics a type $\theta'$ while in choosing its quality $q''_i$, the firm mimics a potentially different type $\theta''$. %
Second, we will discard all the IC constraints \eqref{IC} for $\theta_i'\neq \theta_i''$. That is, we will focus only on deviations whereby the firm mimics the same type in its message and quality decisions.
Finally, by using the Envelope Theorem, we will discard all ``distant'' IC constraints, leaving in place only the local ones. 

Define the equilibrium \textbf{outcome functions} induced by $\mathcal{M}$ and $\bm\sigma$ as $$\bm z(\bm \theta) := \bm z(\bm\sigma(\bm\theta)), \quad \bm t(\bm\theta) := \bm t(\bm\sigma(\bm\theta)).$$ Note that these functions are different from a direct, truthful mechanism, which entails a pair of mappings $z(\theta,q)$, $t(\theta,q)$.\footnote{The reason is that in such a direct mechanism for every (potentially untruthful) report $\theta'_i$ the buyer would \emph{impose} the quality $q_i(\theta'_i)$ on the firm $i$. However, he can not impose qualities, as they are chosen beforehand; the designer can only impose outcomes $z,t$ conditional on qualities observed, but not the qualities themselves} After making the first two relaxations (setting $\theta'_i=\theta''_i$), 
the buyer's problem becomes
\begin{align*}\nonumber
\hypertarget{P1}{\mbox{(P1) }} & \max\limits_{q_i(\theta_i),\bm z(\bm\theta),\bm t(\bm\theta)} \mathbb{E}U(\bm z(\bm\theta),\bm t(\bm\theta),q_1(\theta_1),\ldots,q_n(\theta_n))\\
\nonumber
\mbox{s.t. } & \mathbb{E}_{\theta_{-i}}\Pi_i(z_i(\bm\theta),t_i(\bm\theta),q_i(\theta_i),\theta_i)\geq \\
& \mathbb{E}_{\theta_{-i}}\Pi_i(z_i(\theta'_i,\theta_{-i}),t_i(\theta'_i,\theta_{-i}),q_i(\theta'_i),\theta_i),\\
& \mathbb{E}_{\theta_{-i}}\Pi_i(z_i(\bm\theta),t_i(\bm\theta),q_i(\theta_i),\theta_i)\geq 0,\\
& z_i(\bm\theta)\geqslant 0, \ \sum_{i=1}^n z_i(\bm\theta)=1, \ q_i(\theta_i) \geq 0
\end{align*}
for all $i$ and $\theta'_i, \bm \theta$ in the support.
Thus, we have replaced optimization over (indirect) mechanisms and equilibria with optimization over the outcome functions.

We employ the standard envelope argument that uses local IC constraints to express interim expected transfers $t_i$ as a function of allocation $z_i$ and quality $q_i$, with the firm's profit at the worst-off type normalized to zero.  We then use standard integration by parts to arrive at the following problem:
\begin{align}
\hypertarget{P2}{\mbox{(P2) }} & \max\limits_{q_i(\theta_i),\bm z(\bm\theta)} \mathbb{E} \sum_{i=1}^n \left(V(q_i)z_i - \tilde C^{P}(q_i, \theta_i)z_i - \tilde C^I(q_i, \theta_i)\right) \\
\mbox{s.t. } & z_i(\bm\theta)\geqslant 0, \ \sum_{i=1}^n z_i(\bm\theta)=1, \ q_i(\theta_i) \geq 0\label{zconstraints}
\end{align}
where $\tilde{C}^P(q_i, \theta_i)$ are the $i$'th firm's \textbf{virtual production costs} defined as
\begin{equation*}
\tilde C^P(q, \theta) := C^P(q, \theta) + C^{P}_{\theta}(q, \theta)\frac{F(\theta)}{f(\theta)}
\end{equation*}
and $\tilde{C}^I(q_i, \theta_i)$ are the $i$'th firm's \textbf{virtual investment costs} defined as
\begin{equation*}\tilde C^I(q, \theta) := C^I(q, \theta) + C^{I}_{\theta}(q, \theta)\frac{F(\theta)}{f(\theta)}.
\end{equation*}

Our aim is to solve the relaxed problem \hyperlink{P2}{(P2)}, getting a certain optimal outcome functions $\bm z(\bm\theta)$ and $q_i(\theta_i)$ and an upper bound $\overline{U}$ on the buyer's utility from any mechanism. After that we will show that, under fairly standard regularity conditions, $\overline{U}$ will be attained in an equilibrium of an actual mechanism implementing $\bm z(\bm\theta), \bm t(\bm\theta)$ and  $q_i(\theta_i)$, thus proving that this solution is indeed optimal in the full problem \hyperlink{P}{(P)}. 

The appropriate regularity conditions are formulated below.

\begin{stand}\label{regularity} \noindent%
\begin{enumerate}
\item $\tilde{C}^P_{\theta}>0$, $\tilde{C}^P_{qq}>0$, $\tilde{C}^P_{q\theta}\geq 0$, $\tilde{C}^I_{\theta}>0$, $\tilde{C}^I_{qq}>0$, $\tilde{C}^I_{q\theta}\geq 0$.
    \item $V'(q) > \tilde{C}^P_{q}(q,\theta)$ for $q = 0$, and $V'(q)<\tilde{C}^P_{q}(q,\theta)+\tilde{C}^I_{q}(q,\theta)$ for $q = \infty$ %
\end{enumerate}

\end{stand}

\setcounter{assu}{2}

{The latter Inada-style conditions are only needed for the interiority of the solutions.} By default, we will assume \cref{cost_properties,regularity} throughout the paper.

We will consider separately two cases: (i) the buyer restricts himself to using only symmetric mechanisms, leading to an upper bound of $\overline{U}_{sym}$; (ii) the buyer can possibly employ asymmetric mechanisms, leading to an upper bound of  $\overline{U}_{asy}$. One of the key messages of this paper is that frequently $\overline{U}_{sym}<\overline{U}_{asy}$ even though the firms are ex-ante symmetric.

\section{Optimal symmetric mechanisms}\label{sect:symmetric}
In this section, we characterize optimal symmetric\footnote{Symmetry here means that the quality functions and allocation functions are invariant under the permutations of indices.} mechanisms, similar in spirit to the optimal mechanism in \cite{che1993design}, but in the presence of investment (all-pay) costs. This is done in two steps: first, we find the optimal symmetric outcome functions, second, we show that a symmetric mechanism exists that implements them. See \cref{app:symmetric} for the technical lemmas and proofs.

\subsection{Optimal outcomes}
Define the \textbf{virtual production surplus} as \[x(q,\theta):=V(q)-\tilde{C}^P(q,\theta),\]
then the buyer's payoff can be written as
\begin{equation*}\label{V_b}
U=\mathbb{E}\sum_{i=1}^n \left(x(q_i, \theta_i) z_i(\theta_i)-\tilde{C}^I(q_i,\theta_i)\right).
\end{equation*}

With no investment costs, as in \cite{che1993design}, the optimal quality function would be the one maximizing $x(q,\theta)$ pointwise, whether or not the symmetry constraint is applied. The situation becomes more complicated when investment costs are involved. 

Firstly, for any quality functions $q_i(\theta_i)$, allocating the contract to the firm with the highest $x(q,\theta)$ is optimal, thus $\sum_{i=1}^n x(q_i, \theta_i) z_i(\theta_i) = \max_i x(q_i, \theta_i)$. Secondly, restricting ourselves to symmetric outcomes, we can argue that, under appropriate regularity conditions (\cref{regularity}, part 1), it is the firm with the lowest (efficient) type, which we denote as $\theta_{(1)}$. In other words, the \textbf{optimal symmetric allocation functions} which we denote as $z^*_{i,sym}(\theta)$ are equal to $\mathbb{I}[\theta_i=\theta_{(1)}]$.

{The last step is not obvious, since, generally, $\max_i x(q(\theta_i), \theta_i) \geqslant x(q(\min_i \theta_i), \min_i \theta_i)$, unless $x(q(\theta), \theta)$ is shown to be monotone in $\theta$ at the symmetry-constrained optimum. To argue this, we can replace any non-monotone $x(q(\theta), \theta)$ with its equimeasurable non-increasing rearrangement, yielding an increase in the payoff, see \cref{x-decreasing}.}

Consequently, we can reduce attention to the following payoff function
\[U_{sym}:=\mathbb{E}\left(x(q(\theta_{(1)}),\theta_{(1)})-n \tilde{C}^I(q(\theta_i),\theta_i)\right).\]

Define the \textbf{symmetric probability of winning} given type $\theta$ as
$$\textit{PW}_{sym}(\theta) := (1-F(\theta))^{n-1},$$
then the buyer's payoff is equal to
\begin{equation}\label{int}
	U_{sym}=n\int \left[x(q(\theta),\theta) \cdot \textit{PW}_{sym}(\theta)-\tilde{C}^I(q(\theta),\theta)\right]f(\theta)d\theta,
\end{equation}
due to the pdf of $\theta_{(1)}$ being equal to  $n(1-F(\theta))^{n-1}f(\theta)$. 

The \textbf{optimal symmetric quality function} solving \hyperlink{(P2)}{(P2)}, which we denote as $q^*_{sym}(\theta)$, can then be easily derived via first order conditions since it maximizes \eqref{int} pointwise. It remains to check that it is decreasing, which follows from \cref{regularity}, part 1.

We summarise the derivation in the proposition below.

\begin{prop}\label{optimalsym}
The quality function $q^*_{sym}(\theta_i)$ solving the relaxed problem (P2) under the symmetry constraint is determined by finding $q$ that solves the equation below
\begin{equation}\label{virtual_surplus}
V'(q) = \tilde{C}^P_q(q,\theta)) + \tilde{C}^I_q(q,\theta)/\textit{PW}_{sym}(\theta).   
\end{equation}
Moreover, $q^*_{sym}(\theta_i)$ is decreasing, and $z^*_{i,sym}(\bm \theta) = \mathbb{I}[\theta_i=\theta_{(1)}].$
\end{prop}

Notice that the trade-off between the virtual production surplus and investment costs is steeper for the firms with higher $\theta_i$ because they win less frequently. In particular, the highest (least efficient) type $\theta_i = 1$ in our model produces the worst possible quality $q^*_{sym}(1) = 0$, due to the monotonicity of the investment costs in $q_i$, see \cref{cost_properties}.

Substituting $q^*_{sym}(\theta_i)$ into $U_{sym}$ yields an upper bound $\overline U_{sym}$ on the buyer's utility in the full problem \hyperlink{P}{(P)}. It remains to show that it can also be attained.

\subsection{Implementation}
\label{sec:symmimplement}

In this section, we will show that the optimal quality function $q^*_{sym}(\theta_i)$ and the associated allocation functions $z_{i,sym}^{*}(\bm\theta)$ are implemented in the equilibrium of a first-score auction, with a carefully picked scoring rule\footnote{Formally, a first-score auction is a mechanism such that, for every firm $i$, the message set is $M_i=\mathbb{R}_+$ with message being the price offered. Let $W(\bm m,\bm q)$ be set of auction winners given by $W(\bm m,\bm q)=\{i|s(q_i)-m_i=\max_j(s(q_j)-m_j)\}$. Firm's $i$ allocation is then given by $z_i(\bm m,\bm q)=\mathbb{I}(i\in W(\bm m,\bm q))/|W(\bm m,\bm q)|$. Firm's $i$ transfer is $t_i(\bm m,\bm q)=m_i z_i(\bm m,\bm q)$.}. In fact, we will show that this is true for all decreasing quality functions $q(\theta)$ and the same allocation functions. This will also imply that none of the constraints in problem \hyperlink{P}{(P)} are violated, and thus $\overline U_{sym}$ is indeed attained.

As is usual in the quasi-linear environments, we shall seek a \textbf{quasi-linear scoring rule} $S(q,p)=s(q)-p$, where $p$ is the price-bid. It is not hard to show that $s(q)$ must, due to every firm's first-order condition, satisfy
\begin{equation}\label{scoredef}
s'(q) =  C^P_q(q,\theta(q))+C^I_q(q,\theta(q))/\textit{PW}_{sym}(\theta(q)),
\end{equation}
where $\theta(q)$ is the inverse of the chosen quality function $q(\theta)$.

This leads to the following equilibrium score and price-bid strategies
\begin{align}
S(\theta) & =s(q(\theta))-p(\theta),\label{eqstrat1}\\     
p(\theta) & =C^P(q(\theta),\theta)+\frac{C^I(q(\theta),\theta)+\textit{IR}(\theta,1)}{\textit{PW}_{sym}(\theta)}.\label{eqstrat2} 
\end{align}
where $\textit{IR}(\theta, \theta')$ are the \textbf{informational rents} accumulated between types $\theta$ and $\theta'$%
\begin{gather}\label{IR1}
\textit{IR}(\theta, \theta') := \int_{\theta}^{\theta'}[C^P_{\theta}(q(u),u)\textit{PW}_{sym}(u)+C^I_{\theta}(q(u),u)]du.
\end{gather}

However, it is not clear that global IC constraints will hold as we must exclude all joint quality-price deviations.

We illustrate the problem with a stylized plot of the agent's  response curves $\hat{q}(S|\theta)$ and $\hat{S}(q|\theta)$, for a fixed type $\theta$ and the equilibrium score distribution of opponents. When there were no investment costs, as in \cite{che1993design}, quality was chosen independently from the score. Thus, it was clear that $\hat{S}(q|\theta)$ was crossing $\hat{q}(S|\theta)$ only once and also ``from below'', see \cref{fig:doubledev} (left), which characterizes the intersection point as a global maximum. To the contrary, in the presence of investment costs, $\hat{q}(S|\theta)$ is not constant, and so the number of intersections is not obvious, see \cref{fig:doubledev} (right). 

But even assuming a unique intersection and the monotonicity of both response curves, we still have to argue that $\hat{S}(q|\theta)$ crosses $\hat{q}(S|\theta)$ ``from below''. To make things worse, by deviating from the conjectured optimum in the score dimension only, it is impossible to distinguish the type of crossing. Thus, considering (double) deviations in score and quality are absolutely necessary for the proof of optimality.

\begin{figure}[t!]
\centering
\includegraphics[scale=0.9,width = .48\linewidth]{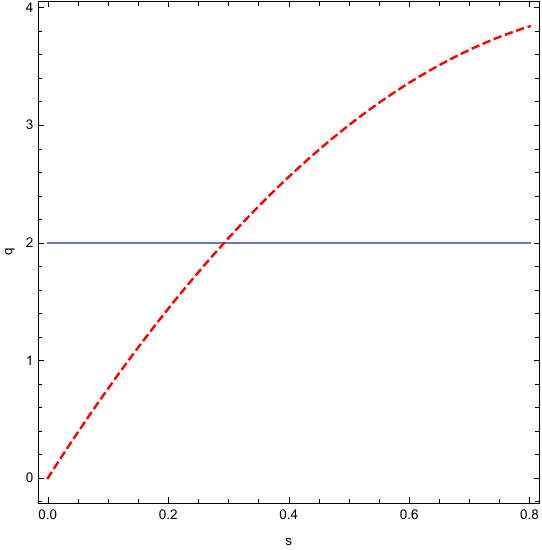}
\includegraphics[scale=0.9,width = .48\linewidth]{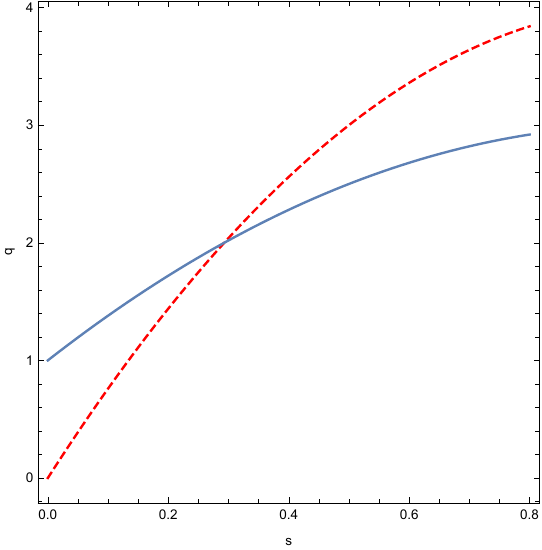}
\caption{Stylized plot of the optimal response curves $\hat{q}(S|\theta)$ (solid line) and $\hat{S}(q|\theta)$ (dashed line) without investment costs (left) and with investment costs (right).}
\label{fig:doubledev}
\end{figure}

Despite this complication, no condition beyond \cref{cost_properties} is needed to guarantee that the implementation works. The trick is to decompose the deviation towards an arbitrary point $(q,S)$ into a (sideways) deviation along the $q=\hat{q}(S|\theta)$ curve, and a (vertical) deviation in quality, see \cref{fig:doubledev}. It remains to verify that the $q=\hat{q}(S|\theta)$ curve passes through the conjectured score-quality pair and that the agent would choose it among all other points on that curve. In other words, we need to confirm that $q(\theta)$ is optimal along the $S = S(\theta)$ line and that $S(\theta)$ is optimal along the $q=\hat{q}(S|\theta)$ curve, for each $\theta$. These two conditions amount to $S(\theta), q(\theta)$ being a proper best response.\footnote{We stress that $\hat{q}(S|\theta)$ different from the curve $(q(\theta), S(\theta))$, which is not depicted on \cref{fig:doubledev}. Thus, ruling out deviations associated with lying about ones type, even if we chose to do so, would not make the analysis easier.}

{With the double deviations taken care of, we can establish the existence of an equilibrium in a first-score auction that implements a desired outcome.}

\begin{prop}\label{symmimplement}
For any decreasing quality function $q(\theta)$ with inverse $\theta(q)$, the first-score auction with quasi-linear score $S(q,p)=s(q)-p$, where $S(q)$ satisfies \eqref{scoredef}, has a BNE in which the strategy of every firm is $q(\theta)$ and $p(\theta)$, $S(\theta)$ defined by formulas \eqref{eqstrat1}, \eqref{eqstrat2}, the most efficient firm gets the contract, and the expected profit of every firm type is non-negative.
\end{prop}

Finally, by applying \cref{symmimplement} to the optimal quality function $q^*_{sym}(\theta)$ derived in \cref{optimalsym}, we can claim or first main result~--- a full characterization of the optimal symmetric mechanism.

\begin{thm}\label{th1}
The first-score auction with the score $S(q,p)=S^*(q)-p$, where $S^*(q)$ satisfies \eqref{scoredef} with $\theta(q)$ being the inverse of $q^*_{sym}(\theta)$ solves the problem \hyperlink{(P)}{(P)} with the restriction to symmetric mechanisms, i.e., $\bar U_{sym}$ is attained, and thus the first-score auction is an optimal symmetric mechanism.
\end{thm}

It is important to note that since the agents can signal their one-dimensional type through either quality or price, there are alternative implementations of these outcomes. In particular, it is possible to write a contract where agents compete only in the quality dimension, effectively choosing from a menu of price-quality pairs. The sets of messages $M_i$ will be empty in such a mechanism. However, we believe that the specific implementation via a scoring auction would be of greater value for practitioners.\footnote{In many countries, buyers are legally bound to use the scoring auction as a procurement mechanism in environments where quality matters. One could also argue that letting the agents compete in both price and quality generates more information relative to the price-quality menu system and thus may improve price discovery.}

\subsection{Comparative statics}

Having characterized the optimal symmetric quality function to implement and the optimal score function, we are now able to investigate how the solution in our model differs from that of \cite{che1993design}, without investment costs.

To introduce the size of the investment costs into the model, we simply scale $C^I(q,\theta)$ by a parameter  $\beta \geqslant 0$. Note that $\beta=0$ corresponds to the model of \cite{che1993design}. As usual, we will denote the number of firms as $n$.

\subsubsection{Downward quality distortion}

One of the main findings in \cite{che1993design} was that because naive scoring rules fail to take the informational rents into account, it becomes optimal to systematically discriminate against quality by lowering the slope of the scoring rule. As a result,  quality is distorted downwards, but the information rents of the relatively efficient firms are also reduced. The same intuition applies here.

Subtracting \eqref{virtual_surplus} from \eqref{scoredef} and using the definitions of the virtual costs, one gets
\begin{equation}\label{score_gap}
s^{*\prime}(q)-V'(q)=-C^P_{q\theta}(q,\theta)\frac{F(\theta)}{f(\theta)}-\frac{C^I_{q\theta}(q,\theta)}{\textit{PW}_{sym}(\theta)}\frac{F(\theta)}{f(\theta)}\leqslant 0,
\end{equation}
where the inequality follows from \cref{regularity}.
\begin{prop}\label{downward}
For the optimal scoring rule, $s^{*\prime}(q)\leqslant V'(q)$. 
\end{prop}

Interestingly, it also follows from \eqref{score_gap} that when the production costs depend only on type, and the investment costs depend only on quality, the truthful score with $S(q,p)=V(q)-p$ is actually optimal; no score distortion is needed. This observation will be helpful in \cref{sect:asym} where we will focus on optimal asymmetric mechanisms in such environments.

\subsubsection{Number of firms and the size of investment costs}

Another important takeaway from \cite{che1993design} was that neither the quality of the firm nor the scoring rule in the optimal mechanism depended on the number of firms $n$. That is because without the investment costs, it was a dominant strategy for each firm to choose quality to maximize it's own strength (relative to the scoring rule) in the upcoming scoring auction. Furthermore, since the optimal scoring rule was meant to induce the same quality function for each firm, it also did not depend on $n$. It turns out that neither is generally true in the presence of investment costs.

Recall that $q^*_{sym}(\theta)$ maximizes 
\begin{equation}\label{symm_obj}
x(q,\theta)-\frac{\beta}{\textit{PW}_{sym}(\theta)} \tilde C^I(q,\theta)    
\end{equation}
pointwise. For any given type $\theta_i$, a higher number of firms lowers the perceived probability of winning in the symmetric mechanism. As a result, the trade-off between the virtual production surplus and the virtual investment costs, captured by the formula \eqref{symm_obj}, becomes steeper, thus inducing lower quality. The size of investment costs $\beta$ acts in a similar way. \footnote{Formally, the function \eqref{symm_obj} is supermodular in $(-q,\beta,n)$, thus, the optimal symmetric quality function is decreasing in $\beta, n$.} 

\begin{prop}\label{q_in_n}
The optimal symmetric quality function $q^*(\theta)$ is decreasing in $n, \beta$.
\label{symmcompare}
\end{prop}

The dependence of the scoring rule on $n, \beta$ is less obvious. Recall that the optimal scoring auction is defined by the following first-order condition
\begin{equation*}
s'(q) =  C^P_q(q,\hat \theta)+\frac{\beta}{\textit{PW}_{sym}(\hat \theta)}C^I_q(q,\hat \theta),
\end{equation*}
where $\hat \theta=\theta(q)$ is the inverse of $q^*_{sym}(\theta)$. On the one hand, higher $n, \beta$ increase the right-hand side of this equation for any given $\hat \theta$. On the other hand, $q^*_{sym}(\theta)$ is decreasing in $n, \beta, \theta$ and thus $\theta(q)$ is also decreasing in $n, \beta$, for any given $q$, which decreases the right-hand-side of the equation. Thus, the cumulative effect is ambiguous.

To gain traction, we restrict attention to environments in which the elasticity of costs and $F(\theta)$ w.r.t $\theta$ is constant. These elasticities can be interpreted as the heterogeneity of firms or the importance of private information with respect to these costs.

\begin{prop}\label{elmatters}
Suppose $F(\theta)=\theta^{\frac{1}{\delta}}$ for some $\delta>0$, $C^P(q,\theta)=\theta^{E_P }g_1(q), C^I(q,\theta)=\beta\cdot  \theta^{E_I } g_2(q)$ where $E_P ,E_I >0$ and $g_1(q)$, $g_2(q)$ are some well-behaved functions, then:
\begin{enumerate}
\item  If $E_P <E_I $ then $s^{*\prime}(q)$ decreases in $n, \beta$ at every $q$,
\item If $E_P >E_I $ then $s^{*\prime}(q)$ increases in $n, \beta$ at every $q$.
\end{enumerate}
\end{prop}

Interestingly, when the elasticities are equal $E_P = E_I = E$, the optimal scoring rule does not depend on either $n$ or $\beta$, and can be, in fact, written out: $S(q, p) = \frac{V(q)}{1+\delta E} - p$.

A practical implication of \cref{elmatters} is that the buyer should adapt the quality incentives to the degree of competition, i.e., the number of firms $n$. If there is more heterogeneity in production costs than in investment costs $(E_P >E_I)$, the quality incentives should become more high-powered when competition rises. If there is less heterogeneity in production costs than in investment costs $(E_P < E_I)$, the quality incentives should become less high-powered when competition rises.\footnote{If the buyer is using a linear scoring rule $S(q,p) = \alpha q - p$, a simple increase in $\alpha$ would make the quality incentives more high-powered.} 

The parameter $\beta$ can be interpreted as either a tightening of regulation that increases the firms' investment costs or as the inverse of the number of contracts in which the firm simultaneously competes. In either case, an increase in $\beta$ acts similarly to the increase in $n$.

\section{Suboptimality of symmetric mechanisms}
\label{sec:suboptimal}

In this section, we show that the solution to our mechanism design problem is generically not a symmetric mechanism. We identify two cases, where the optimal mechanism is asymmetric with certainty. See \cref{app:suboptimal} for the formal proofs.

\subsection{Vanishing private information}
\label{sec:vanish}

The first case is that of vanishing private information. We introduce the latter by parametrizing the production and investment costs with $\alpha \geqslant 0$ like $C^P(q,\alpha \theta)$ and $C^I(q,\alpha \theta)$, so that when $\alpha$ vanishes, the costs become independent on the realization of the private type. We will refer to $\alpha$ as the \textbf{importance of private information}.

Intuitively, if private information does not enter in any of the costs, the optimal thing to do is always to award the contract to one ex-ante chosen firm so that the others do not incur wasteful investment costs. However, there is a caveat. The principal might benefit from screening if $x_i(\theta_i)$ are not constant. 

Our result is that the optimal symmetric mechanism is still sub-optimal, when private information is vanishing.

\begin{prop}\label{smallalpharesult}
There exists $\overline{\alpha}>0$ such that for all $0 \leqslant \alpha < \overline{\alpha}$ the optimal symmetric mechanism is not an optimal mechanism.
\end{prop}

Above we have registered a particular kind of dominance --- by a single-bidder mechanism, or sole-sourcing. However, there can be other mechanisms that may dominate the optimal symmetric mechanism.

\subsection{Elastic indirect investment costs}

For the second case, we need to introduce a reduction of the problem to an equivalent problem with investment costs only. 

Define the \textbf{indirect investment costs function} $C(x,\theta)$ as $$C(x,\theta):=\min_q\{\tilde{C}^I(q,\theta):V(q)-\tilde{C}^P(q,\alpha \theta)=x\},$$ 
for all $x$ in the range of $V(q)-\tilde{C}^P(q,\alpha \theta)$. Given that quality is one-dimensional, and $\tilde C^I(q,\theta)$ is monotone in $q$, the minimum is achieved at the lowest root $q$ of the equation $V(q)-\tilde{C}^P(q,\alpha \theta)=x$.\footnote{In \cref{C-properties} in the Appendix we show that given \cref{regularity}, $C(x,\theta)$ is increasing in both arguments and supermodular.} After this cost minimization step, the principal's utility may be written simply as
\begin{equation}\label{utility_x}
U=\mathbb{E}(\max_i x_i(\theta_i)-\sum_{i=1}^n C(x_i(\theta_i),\theta_i)). 
\end{equation}

The problem \hyperlink{P2}{(P2)} is equivalent to choosing functions $x_i(\theta_i)$ to maximize \eqref{utility_x} and thus depends solely on the shape of the indirect investment costs. We would like to make these costs relatively elastic. 

We will first observe suboptimality of symmetric mechanisms in a special case: $V(q) = q, C^P(q,\alpha\theta) = \alpha \theta, C^I = q^\gamma/\gamma$ and $F(\theta) = \theta$ thus optimal quality is equal to $x + 2\alpha \theta$ and $C(x,\theta) = (x + 2\alpha \theta)^{\gamma}/\gamma$, where $\gamma$ captures the elasticity.

The optimal symmetric mechanism, derived in \cref{sect:symmetric}, is characterized by $$x^*_{sym}(\theta) = (1-\theta)^{\frac{1}{\gamma-1}}-2\alpha \theta.$$
For any cutoff $\theta_0 \in [0,1]$, define a pair of functions
\begin{equation}\label{x-example}
x_1(\theta) = (1-\min(\theta,\theta_0))^{\frac{1}{\gamma-1}}-2\alpha \theta , \quad x_2(\theta) = \mathbb{I}(\theta<\theta_0)(1-\theta)^{\frac{1}{\gamma-1}}-2\alpha \theta,    
\end{equation}
 see \cref{fig:superexample} (left). These functions are feasible in the relaxed problem \hyperlink{P2}{(P2)} and thus can be used to establish dominance there. These functions are associated with right-censored outcome functions which we define later in \cref{sec:optimaloutcomes} and  implement via a score floor auction in \cref{sec:implementation}. Thus these functions are also feasible in the full problem \hyperlink{P}{(P)}.

For $\gamma = 3$ and $\alpha = 4/10$, the utility of the buyer as a function of the cutoff is
$$ U(\theta_0) = \frac{2}{15} \left(\theta_0^3+\left(\sqrt{1-\theta_0}-3\right)
   \theta_0^2+\left(3-2 \sqrt{1-\theta_0}\right)
   \theta_0+\sqrt{1-\theta_0}+1\right).$$
This function is maximized at an interior point $\theta_0 = 11/36$, see \cref{fig:superexample}, yielding an approximately 3\% increase in utility relative to the extreme points (sole sourcing at $\theta_0 = 0$ and optimal symmetric mechanism at $\theta_0 = 1$) that have the same utility in this example. Thus, an optimal mechanism is asymmetric.

\begin{figure}[t!]
\centering
\includegraphics[width = .45\linewidth]{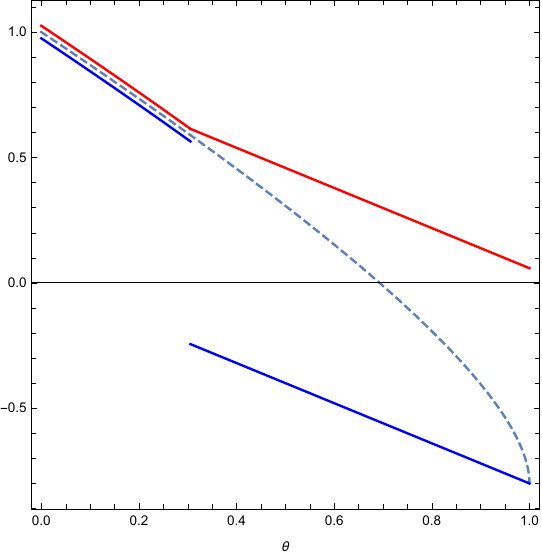}
\includegraphics[width = .45\linewidth]{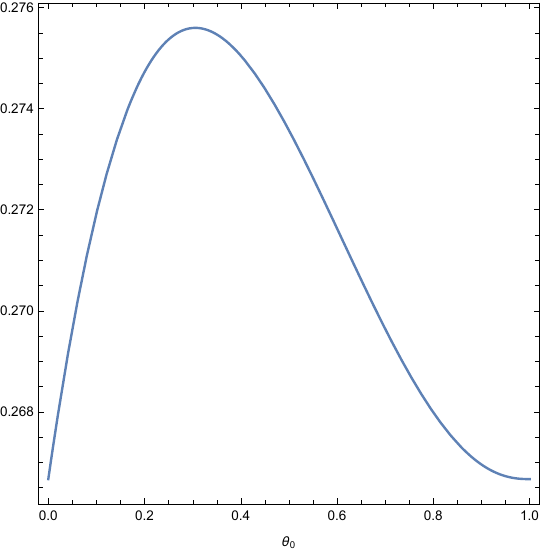}
\caption{$x_1(\theta),x_2(\theta),x^*_{sym}(\theta)$ functions (left) with the optimal cutoff $\theta_0 = 11/36$ and the associated buyer's utility (right) as function of the cutoff $\theta_0$, for $\gamma = 3,\alpha = 4/10$.}
\label{fig:superexample}
\end{figure}

More generally, let $\gamma$ denote the unique number such that $\lim\limits_{x\to\underline{x}^+}\frac{C(x,1)}{(x-\underline{x})^{\gamma}}$ is positive and finite, where $\underline{x}=V(0) - \tilde C^P(0,\alpha)$. \footnote{$\underline{x}$ is the solution to equation $0=C(\underline{x},1)$.} That is, for small $x$ the indirect cost function behaves like a power law with degree $\gamma$. For this setting, we show that, independently of $\alpha$, the optimal symmetric mechanism is dominated when $\gamma$ is large enough. This is valid not only for the constant-elasticity case as in the example above.

\begin{prop}\label{symmetricbad2}
For $n = 2$, for all $\gamma > 2$, the optimal symmetric mechanism is not an optimal mechanism.
\end{prop}

\section{Optimal asymmetric mechanisms for two firms}
\label{sect:asym}

In \cref{sec:suboptimal}, we have seen that symmetric mechanisms are not, in general, optimal for procurement in the presence of investment costs. What would then be an optimal mechanism? In this section, we provide results about optimal mechanisms without the \emph{a priori} restriction to the symmetric ones. See \cref{app:asym} for the formal proofs.

In general, the analysis is hard; to the best of our knowledge, only partial results are available in the literature for similar problems\footnote{\cite{zhang2017auctions} provides the result for $n=2$ and additively separable environment with quadratic costs; \cite{gershkov2021theory} provide a general condition under which the optimal mechanism is symmetric; \cite{celik2009optimal} study binary investment decisions.}.
The key difficulty comes from the assumption that a firm's quality is allowed to depend only on her own type. This precludes the use of pointwise integral maximization, which is typical in mechanism design. As a result, one has to use variational techniques to characterize an optimum. %

We are able to fully characterize the optimal mechanism under additional restrictions to $n=2$ players and the separable cost structure that we define below. 

\begin{assu}\label{separablestructure}
Separable cost structure:
\begin{enumerate}
\item $C^P(q,\theta)=\alpha\theta$ for some $\alpha>0$, 
\item $C^I(q,\theta)=g(q)$, for some function $g(q)$: $g'(q)>0$, $g''(q)>0$.
\item $V(q) = q$
\end{enumerate}
\end{assu}

Assumption $V(q) = q$ does not lead to the loss of generality, as quality can always be reparametrized. The $\alpha$ parameter here is consistent with the one defined in \cref{sec:vanish}.

We start by defining, for the set-up just described, the notion of convexity that, as we show, is relevant for the shape of the optimal mechanism. Recall that $F(\theta)$ is the cdf of the type $\theta$.
Denote by $J(\theta)=\theta+\frac{F(\theta)}{f(\theta)}$ the standard virtual type. It follows from our \cref{regularity} that $J(\theta)$ is increasing. Define $\xi(z):=1-J(F^{-1}(1-z))$. As $J(\theta)$ is increasing, $\xi(z)$ is also increasing. 

We say that \textbf{marginal investment costs are sufficiently convex} iff the function $q\to \alpha \xi(g'(q))-q$ is strictly quasi-convex. We say that \textbf{marginal investment costs are sufficiently concave} iff the function $q\to \alpha \xi(g'(q))-q$ is strictly quasi-concave.

Note that the above definition depends not only on the marginal investment costs $g'(q)$ but on the other primitives as well, thus, the qualifier ``sufficiently'' is context-dependent. In the simple setting where the type $\theta$ is distributed uniformly on [0,1], $\alpha\xi(g'(q))-q=2\alpha g'(q)-q-1$, so any convex marginal costs $g'(q)$ would be sufficiently convex while any concave marginal costs $g'(q)$ would be sufficiently concave. 

\subsection{Optimal outcomes}
\label{sec:optimaloutcomes}

Recall that in \cref{sec:suboptimal}, we showed that the optimal symmetric mechanism could be dominated by a certain asymmetric mechanism, which can be thought of as an optimal symmetric mechanism with a special (asymmetric) treatment of high (inefficient) types. 

In a restricted setting with two firms, we are able to prove that this asymmetric mechanism will be, in fact, optimal under our novel condition. Recall the optimal symmetric quality functions $q^*_{sym}(\theta)$ and denote a optimal pair of quality functions as $(q^*_1(\theta), q^*_2(\theta))$ and the optimal pair of allocation functions as $(z^*_1(\bm\theta),z^*_2(\bm\theta))$.

\begin{figure}[t!]
\centering
\includegraphics[width = .45\linewidth]{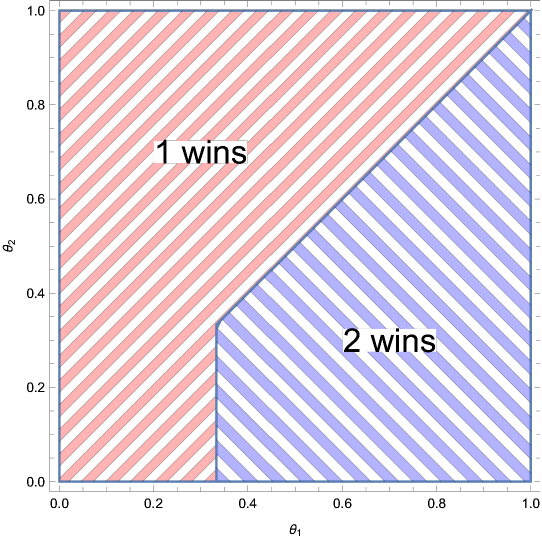}
\includegraphics[width = .45\linewidth]{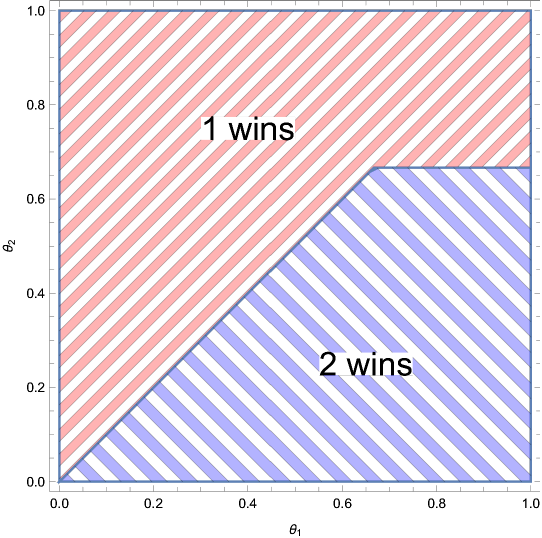}
\caption{Allocation function for the (favored) firm 1 and (unfavored) firm 2 in the score ceiling (left) and score floor (right).}
\label{fig:figureone}
\end{figure}

When the marginal costs are sufficiently convex, we show that the relaxed problem \hyperlink{(P2)}{(P2)} is solved by the \textbf{right-censored allocation functions}: 
\begin{gather}\label{right_allocation}
z^*_1(\bm\theta) = \mathbb{I}(\tilde \theta_1 < \theta_2), \quad
z^*_2(\bm\theta) = 1 - z^*_1(\bm\theta),
\end{gather}
where $\tilde \theta_1 = \min(\theta_0,\theta_1)$. In other words, the allocation is efficient, but as if the type of firm 1 was capped at the threshold $\theta_0$, i.e., censored from the right. The associated \textbf{right-censored quality functions} are defined as below:
\begin{gather}\label{q_floor}
q^*_1(\theta_1) = q^*_{sym}(\tilde \theta_1), \quad 
q^*_2(\theta_2) = \begin{cases}
q^*_{sym}(\theta_2), & \theta_2 < \theta_0\\
q^*_{sym}(1), & \theta_2 > \theta_0.
\end{cases}
\end{gather}
Here, firm 1 produces optimal symmetric quality at her censored type. Regarding firm 2, if her type is above the threshold, she exits, which amounts to producing the lowest possible quality. For the efficient types, i.e., below the threshold, both firms behave like in the optimal symmetric mechanism. We will later implement the right-censored outcome functions using a novel score floor auction.

When the marginal costs are sufficiently concave, we show that the relaxed problem \hyperlink{(P2)}{(P2)} is solved by the \textbf{left-censored allocation functions}: 
\begin{gather}
z^*_1(\bm\theta) = 1 - z^*_2(\bm\theta), \quad 
z^*_2(\bm\theta) = \mathbb{I}(\tilde \theta_2 < \theta_1)
\end{gather}
where $\tilde \theta_2 = \max(\theta_0,\theta_2)$. In other words, the allocation is efficient, but as if the type of firm 2 was censored from the left. The associated \textbf{left-censored quality functions} are defined as below:
\begin{gather}\label{q_ceiling}
q^*_1(\theta_1) = \begin{cases} q^*_{sym}(\theta_1), & \theta_1 > \theta_0\\
q^*_{sym}(0), & \theta_1 < \theta_0
\end{cases}, \quad
q^*_2(\theta_2) = q^*_{sym}(\tilde \theta_2)
\end{gather}
Here, firm 2 produces optimal symmetric quality at her censored type. Regarding firm 1, if her type is below the threshold, she becomes the sole supplier, producing quality at the optimal monopolistic level $q_{sym}^*(0)$. For the inefficient types, i.e., above the threshold, both firms behave like in the optimal symmetric mechanism. 

\begin{figure}[t!]
\centering
\includegraphics[scale=0.9,width = .48\linewidth]{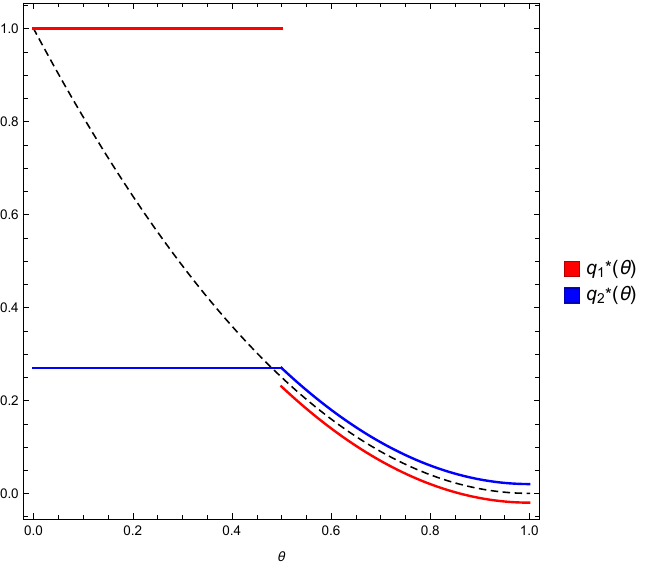}
\includegraphics[scale=0.9,width = .402\linewidth]{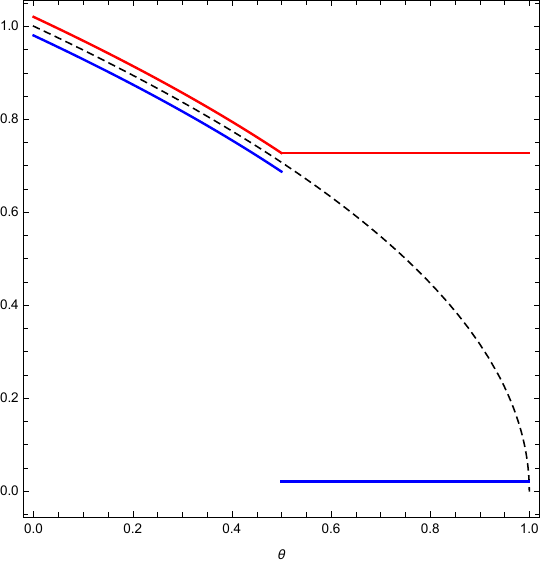}
\caption{Quality function for the (favored) firm 1 and (unfavored) firm 2 in the score ceiling (left) and score floor (right), given separable cost structure with $g(q)=q^{\gamma}/\gamma$.}
\label{fig:figuretwo}
\end{figure}

In both cases, firm 1 is the ``favored'' firm while firm 2 is the ``unfavored'' one, but the censoring is applied to different firms. The typical optimal allocation and quality functions are illustrated in \cref{fig:figureone,fig:figuretwo} for some $\theta_0 \in [0,1]$. Left figure corresponds to left-censored outcomes (score ceilings), right figure corresponds to right-censored outcomes (score floors). 

We summarise in the proposition below.

\begin{prop}\label{main result}
Suppose $n=2$ and \cref{separablestructure} holds. Then \hyperlink{(P2)}{(P2)} is solved by the right-censored allocation and quality functions when the marginal costs are sufficiently convex and by the left-censored allocation and quality functions when the marginal costs are sufficiently concave.
\end{prop}

The result of \cref{main result} is especially clear when the type distribution is uniform.\footnote{Note that \cref{g'''} implies that quadratic (total) investment costs, which are often assumed, are a knife-edge case with $g'''=0$.} 

\begin{coro}\label{g'''}
Suppose that in the setting of \cref{main result} it additionally holds that $\theta_i$ is distributed uniformly on [0,1]. Then:
\begin{enumerate}
    \item If $g'''(q)>0$, the optimal allocation and quality functions are right-censored.
    \item If $g'''(q)<0$, the optimal allocation and quality functions are left-censored.
\end{enumerate}
\end{coro}

Finally, recall that, without the symmetry restriction, the buyer's payoff is:
\begin{equation*}
U_{asy}=\mathbb{E}\sum_{i=1}^n \left(x(q_i, \theta_i) z_i(\theta)-\tilde{C}^I(q_i,\theta_i)\right).
\end{equation*}
Substituting $(q^*_1(\theta),q^*_2(\theta))$ into $U_{asy}$ yields an upper bound $\overline U_{asy}$ on the buyer's utility in the full problem \hyperlink{P}{(P)}. It remains to show that it can also be attained.

\subsection{Implementation}\label{sec:implementation}
Now, we describe the implementation of the optimal quality functions. Similarly to \cref{sec:symmimplement}, we will show the implementability of all left- and right-censored outcome functions, defined by formulas \eqref{q_floor} and \eqref{q_ceiling} for some threshold $\theta_0 \in [0,1]$, and thus the optimal ones.%

\subsubsection{Score floor auction}
In this section, we aim to implement the optimal outcome and quality functions when they are right-censored.

Recall from \cref{symmimplement} that the optimal symmetric quality $q^*_{sym}(\theta)$ can be implemented by a first-score auction with a quasi-linear score function $S(q,p)=s(q)-p$ where $s(q)$ is defined by \eqref{scoredef}. Furthermore, under \cref{separablestructure}, the score function is truthful in the optimal symmetric mechanism, and the equilibrium score strategy takes the form
\begin{align}
S^*_{sym}(\theta) & = V(q^*_{sym}(\theta))-\alpha\theta-\frac{g(q^*_{sym}(\theta))+\textit{IR}(\theta,1)}{\textit{PW}_{sym}(\theta)}.
\end{align}

We introduce a few new details into the traditional design of the first-score auction. First, there is a level $\underline{S}$ below which the score does not count, i.e., the \textbf{score floor}. Effectively, this means that each firm faces a reserve price of $s(q_i) - \underline{S}$, where $q_i$ is her quality. Submitting below the score floor is equivalent to exiting the auction. Second, the favored firm gets a \textbf{bonus} $B$ if her score surpasses $\underline S$ \emph{regardless of winning}. Third, in the case of a tie, the favored firm wins. Apart from that, the auction proceeds as usual. 

Formally, the mechanism is such that messages are the prices offered by the firms.
\begin{itemize}
    \item For the favored bidder, the allocation and transfers are: 
    \begin{align*}z_1(m_1,m_2,q_1,q_2) & =\mathbb{I}(s(q_1)-m_1\ge \max\{s(q_2)-m_2,\underline{S}\})\\
    t_1(m_1,m_2,q_1,q_2) &= m_1\cdot z_1(m_1,m_2,q_1,q_2)+B\cdot \mathbb{I}(s(q_1)-m_1\ge \underline{S}).
    \end{align*}
    \item For the unfavored bidder, the allocation and transfers are:
    \begin{align*}z_2(m_1,m_2,q_1,q_2) &= \mathbb{I}(s(q_2)-m_2>\max\{s(q_1)-m_1,\underline{S}\})\\
    t_2(m_1,m_2,q_1,q_2) &= m_2\cdot z_2(m_1,m_2,q_1,q_2).\end{align*}
\end{itemize}

We will refer to this mechanism as the \textbf{score floor auction} for brevity.

\begin{prop}\label{implement-1}
Suppose $n=2$ and \cref{separablestructure} holds. For any $\theta_0\in[0,1)$, the right-censored outcome and quality functions can be implemented by a score floor auction with a quasi-linear score \eqref{scoredef}, a score floor and bonus equal to $$ \underline S=S^*_{sym}(\theta_0) + \frac{\textit{IR}(\theta_0,1)}{1-F(\theta_0)}, \quad B=\alpha(1-F(\theta_0))(1-\theta_0),$$ and the equilibrium score strategies $S^*_i(\theta) = \max \{S^*_{sym}(\theta)+\frac{\textit{IR}(\theta_0,1)}{1-F(\theta)},\underline{S}\}$ for both firms.
\end{prop}

Similarly to \cref{symmimplement}, to prove \cref{implement-1}, we take care of (double) deviations in score and quality, paying special attention to the partial asymmetry of the allocation and quality functions.  

What is the purpose of the bonus $B$? The score floor alone would not suffice, as it would cut off all relatively inefficient $(\theta>\theta_0)$ types of \emph{both} firms~--- they would decline to participate due to their maximal profit being negative. The bonus restores the incentives to participate for the inefficient types of the favored firm. 

\begin{thm}\label{th2}
Suppose $n=2$ and \cref{separablestructure} holds. If marginal investment costs are sufficiently convex, there exists $\theta_0\in[0,1]$ such that the score floor auction is an optimal mechanism, i.e., it achieves the bound $\overline{U}_{asy}$ and solves the full problem \hyperlink{(P)}{(P)}.
\end{thm}

\subsubsection{Score ceiling auction}

Similar to how the score floor auction implemented the right-censored outcome and quality functions in the previous section, we aim to implement the left-censored outcome and quality functions with a modified scoring auction.

This new auction design features a \textbf{score ceiling} $\bar{S}$ --- the maximally possible level of the score and a \textbf{kickback} $K$, paid by the favored firm only if it won the contract by bidding a score equal to the score ceiling. Similarly to the score floor auction, the favored firm wins in case of a tie. Thus, bidding a score equal to the score ceiling $\bar{S}$ guarantees a victory for the favored firm; it pays the kickback only in this case of sure victory.

Formally, a score ceiling auction is a mechanism such that messages are the prices offered by the firms.    
\begin{itemize}
\item For the favored bidder, the allocation and transfers are: 
    \begin{align*}z_1(m_1,m_2,q_1,q_2) &= \mathbb{I}(\overline{S}\ge s(q_1)-m_1\ge s(q_2)-m_2)\\
    t_1(m_1,m_2,q_1,q_2) &= m_1\cdot z_1(m_1,m_2,q_1,q_2)-K\cdot \mathbb{I}(\overline{S}=s(q_1)-m_1\ge s(q_2)-m_2).\end{align*}
    \item For the unfavored bidder, the allocation and transfers are:
    \begin{align*}z_2(m_1,m_2,q_1,q_2) &= \mathbb{I}(\overline{S}\ge s(q_2)-m_2>s(q_1)-m_1)\\
    t_2(m_1,m_2,q_1,q_2) &= m_2\cdot z_2(m_1,m_2,q_1,q_2).\end{align*}
\end{itemize}

We will refer to this mechanism as the \textbf{score ceiling auction}.

\begin{prop}\label{implement-2}
Suppose $n=2$ and \cref{separablestructure} holds. For any $\theta_0\in[0,1)$, the left-censored outcome and quality functions can be implemented by a score ceiling auction with a quasi-linear score \eqref{scoredef}, a score ceiling $\bar{S} = S_{sym}^*(\theta_0)$ and kickback 
\begin{eqnarray}
\label{kickback1}
K = \max_q[V(q)-g(q)]-\max_q\left[V(q)-\frac{g(q)}{1-F(\theta_0)}\right]+\alpha\frac{F(\theta_0)}{1-F(\theta_0)}\int_{\theta_0}^1(1-F(u))du.
\end{eqnarray}
and the equilibrium strategies $S^*_i(\theta)=\min\{S^*_{sym}(\theta),\bar{S}\}$ for both firms.
\end{prop}
The name ``kickback'' is warranted since $K>0$. To see this, note that $\max_q[V(q)-g(q)]\geq \max_q\left[V(q)-\frac{g(q)}{1-F(\theta_0)}\right]$ as $1/(1-F(\theta))\geq 1$.

We again stress that the role of the side payment $K$ is to provide correct incentives; it is in no way evidence of corruption, which is absent from our model as the buyer and the auctioneer are one and the same. The kickback is needed to ensure that the types $\theta_1>\theta_0$ of the favored firm do not rush to win for sure with the score $\bar{S}$ and thus, sufficient competition with the types $\theta_2>\theta_0$ of firm 2 in the low-score range is maintained.

\begin{thm}\label{th3}
Suppose $n=2$ and \cref{separablestructure} holds. If marginal investment costs are sufficiently concave, there exists $\theta_0\in[0,1]$ such that the score ceiling auction is an optimal mechanism, i.e., it achieves the bound $\overline{U}_{asy}$ and solves the full problem \hyperlink{(P)}{(P)}.
\end{thm}

\subsection{Example: constant-elasticity investment costs}\label{sec:example}
In the previous section, we have described some of the optimal mechanisms that can emerge for two players. In this section, we will characterize all optimal mechanisms for two players and a 2-dimensional class of settings such that the type distribution $F$ is uniform, $V(q) = q$, $C^P(q,\theta)=\alpha\theta$ and $C^I(q, \theta)=q^{\gamma}/\gamma$ where $\gamma>1$ is the elasticity of investment costs and $\alpha>0$ is the importance of private information.

\begin{figure}[t!]
\centering
\caption{Optimal mechanism with $n=2$, $F(\theta) = \theta$, $C^P=\alpha\theta$, $C^I=q^{\gamma}/\gamma$.}
\includegraphics[scale=0.9,width = .9\linewidth]{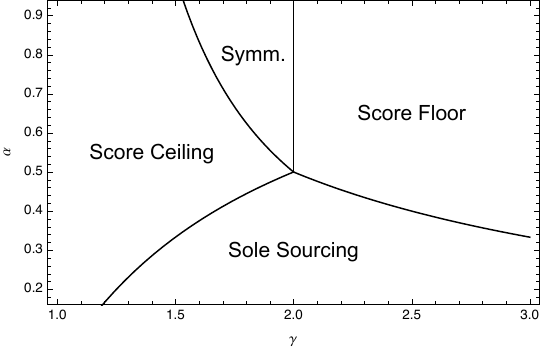}
\label{fig:Ex1}
\end{figure}

The optimal mechanism takes one of four shapes: 
\begin{itemize}
    \item sole sourcing  (optimal 1-bidder mechanism) if $\alpha<\min\left\{1-\frac{1}{\gamma},\frac{1}{\gamma}\right\}$;
    \item a classic scoring auction (optimal symmetric mechanism) if $\gamma< 2$ and $\frac{1}{2(\gamma-1)}<\alpha$; 
    \item a score ceiling auction if $\gamma\leq 2$ and $1-\frac{1}{\gamma}<\alpha<\frac{1}{2(\gamma-1)}$;
    \item a score floor auction if $\gamma>2$ and $\frac{1}{\gamma}<\alpha$.
\end{itemize}
which we illustrate in \cref{fig:Ex1}.\footnote{The locus of parameters for which the optimal mechanism is symmetric can be related to Theorem 2 in \cite{gershkov2021theory}.}

\subsection{Discussion of optimal asymmetric mechanisms}
\subsubsection*{Pooling} The nature of equilibrium supporting our mechanism is that of semi-symmetric and semi-separating. To be precise, the spectrum of types is split by the threshold $\theta_0$. To the one side from the threshold, the equilibrium behaves as a symmetric separating one. On the other side of the threshold, it behaves as an asymmetric (our tie-breaking rule) pooling one. Pooling happens in both score and quality dimensions, but the latter is due to the specific shapes of cost functions in \cref{separablestructure}, and so is not generic. 

We can speculate that for the environments that are neither sufficiently convex nor sufficiently concave, the optimal mechanism would feature ranges of symmetric and asymmetric behavior separated by multiple thresholds and multiple side payments.

\subsubsection*{Side payments} Given the pooling of types at a score floor or ceiling, it is only natural that the transfer experiences a discontinuity (our side payment) at the threshold type to align the IC constraints around the threshold. Without the discontinuity in transfers, the equilibrium utility would be discontinuous due to a sudden change in the probability of winning, so the types next to the threshold would surely deviate. In particular, suppose there was no required kickback in the score ceiling mechanism. In this case, the types of the favored bidder just to the right of the threshold would opt to win without competition, making the outcome more asymmetric than desired. The kickback ensures that the types to the right of the threshold instead choose to compete with the other bidder fairly, thus maintaining exactly the desired degree of outcome symmetry.

It is worth mentioning that score floor and score ceiling mechanisms are not just mirror images of each other. Regardless of the mechanism, every type's equilibrium profit (informational rent) depends on the allocation to all worse types but not better types. This means that the modifications needed to create an asymmetry among the best and among the worst types will be different. To create an asymmetry among the best types without affecting the worst (rightmost) types, one only has to adjust the classical auction design ``for the strong bidders'', so the side payment in the score floor is only paid by the firm that won. A naive attempt to similarly create asymmetry among the worst (rightmost) types would inadvertently interfere with the IC constraints for the best types. Thus, the proper adjustment should affect ``both weak and strong bidders''. As a result, the side payment in the score floor auction, unlike the one in the score ceiling auction, has to be paid regardless of winning. 

\subsubsection*{Ex ante symmetry}

It is also worth mentioning that while the score floor and score ceiling mechanisms feature asymmetric choices of quality and bid, they paradoxically lead to a symmetric distribution of scores. Indeed, the strategies are the same for both favored and unfavored firms for the equilibria suggested in \cref{implement-1} and \cref{implement-2}. Even sole sourcing, the extreme asymmetric mechanism, can be thought of as a special case of either score floor or score ceiling, where two firms arrive with the same (known, constant) score, and one of them wins due to the tie-breaking rule. 

Furthermore, in the range of types where the equilibrium is separating, the profits for the two firms coincide up to a translation constant. Thus, the underlying structure of ex-ante symmetric agents is not entirely lost.

How should the favored bidder be chosen? Of course, the buyer can choose any firm at random. Crucially, the firms should know if they are favored before deciding on their qualities at the announcement of the mechanism. We also recommend not delegating the choice of the favored firm to an independent party, as the latter may want to extort the extra surplus granted by the favored status. To summarize, the choice of the favored firm should be made publicly and credibly.

\subsubsection*{Second-score implementation}

Finally, one may wonder if a dominant-strategy implementation is possible here and whether it would be, in some sense, simpler. Indeed, such an implementation exists when there are no investment costs, see \cite{che1993design,asker2008properties,asker2010procurement}, and is called a second-score auction. Moreover, it has ``truthful bidding'' in the sense that the firms bid so to maximize  the apparent social surplus -- the difference between the buyer's $s(q)$ and firm's $C^P(q)$. Unfortunately, in the presence of investment costs, truthful bidding is not the dominant strategy here, which makes it far less appealing than, for example, truthful bidding in a second-price auction.
\section{Additional results}
In this section, we provide a number of additional results about the asymmetric mechanisms, see \cref{app:additional} for formal proofs.
\label{sec:additional}
\subsection{Symmetry and private information}

We are interested in how the shape of the optimal mechanism changes with the importance of private information $\alpha$.
  
\begin{prop}\label{compstat1}
Suppose $n=2$ and \cref{separablestructure} holds. Then, for the optimal score floor, the optimal threshold $\theta_0^*$ is weakly increasing, while for the optimal score ceiling, it is weakly decreasing in the importance of private information $\alpha$.
Thus, both the optimal score floor mechanism and optimal ceiling mechanism becomes more symmetric as $\alpha$ grows.
\end{prop}

The intuition is related to the main economic trade-off faced by the principal when deciding on the level of mechanism asymmetry. As outlined in the introduction, this is the trade-off between ex-post efficiency and avoidance of duplication of investment costs. On the one hand, the principal would like to award the contract to the most efficient producer to obtain the best possible price-quality combination. However, to find out which one is the most efficient best, one has to treat the firms symmetrically~--- and this will inevitably lead to both firms incurring investment costs. By firms' participation constraints, these costs must ultimately be paid by the principal for both firms. On the other hand, if the principal constrains herself to a sole supplier \emph{a priori}, she will have to compensate only one firm for the investment costs; however, it might not turn out to be the most efficient supplier. In general the principal's dilemma is resolved at some intermediate level of mechanism symmetry. When the importance of private information $\alpha$ grows, the ex-post efficiency motive becomes stronger and the optimal solution moves towards the optimal symmetric mechanism.

Because $\alpha$ is the degree of informational asymmetry between the buyer and the suppliers (the higher $\alpha$, the higher is the degree of the principal's uncertainty about a supplier's production costs), the message of \cref{compstat1} can be formulated as follows: \emph{Informational asymmetry leads to the symmetry of the optimal mechanism}. 
   
\subsection{Symmetry and efficiency}
Even though in this paper we analyze mostly buyer-optimal mechanisms, an interesting question is how the degree of asymmetry compares between the buyer-optimal and the society-optimal (efficient) mechanism. We can easily answer this question when $\theta$ is uniform. %
Somewhat unexpectedly, we find that the efficient mechanism exhibits \emph{more} favoritism than the buyer-optimal mechanism, not less.

\begin{prop}\label{effvsopt}
Suppose $n=2$, $F(\theta) = \theta$ and \cref{separablestructure} holds. Then, for the optimal score floor mechanism, the optimal threshold $\theta_0^*$ is weakly greater, while for the optimal score ceiling, it is weakly lower in the optimal mechanism than in the efficient mechanism.
\end{prop}

The intuition behind \cref{effvsopt} is that while the buyer takes into account virtual costs, the social planner takes into account just costs when determining the optimal degree of mechanism asymmetry. Since virtual costs are more responsive to $\theta$ than just costs (because they also include type-dependent information rents) the buyer is more concerned about ex-post efficiency than the social planner. Hence, by the logic of \cref{compstat1} the buyer should choose a more symmetric mechanism than the social planner. %

\subsection{Restricted entry}
\label{sec:dropk}
One interesting family of asymmetric mechanisms is mechanisms where the principal allows only $k\leq n$ firms to enter, and employs the optimal symmetric mechanism for these $k$ firms. Such mechanisms may be more practical than arbitrary asymmetric mechanisms since this particular kind of asymmetry may be less salient, and thus on the surface such mechanisms may look more ``fair''.

Given that the principal only chooses $k\leq n$, which $k$ would be optimal for him? One reasonable guess is that it may be often optimal to set $k=2$: this choice saves a lot of investment costs while still preserving some competition. We show that in a large class of settings, including the example in \cref{sec:example}, this guess is wrong. Namely, we show that often the optimal solution is ``one-or-all'': depending on the importance of private information, it is always optimal for the principal to either allow only 1 firm or \emph{all} of the firms to enter. In this sizable set of situations, no intermediate $k\in\{2,3,\ldots,n-1\}$ can ever be optimal.

\begin{prop}\label{1orall}
Suppose \cref{separablestructure} holds and $g(0)=0$, $g'(q)$ is strictly increasing and 
$\frac{g(q)}{\sqrt{g'(q)}}$
is also strictly increasing. Suppose also that $F(\theta)=\theta$. Then, the buyer's utility $U(k)$ is a  quasi-convex function of the number of firms $k$ allowed to enter.
Thus, it is optimal for the principal to allow either one or all firms  to enter, that is, $k^*\in\{1,n\}$.
\end{prop}

The condition that $\frac{g(q)}{\sqrt{g'(q)}}$ is strictly increasing means that, in a certain sense, investment costs are not too convex. However, this condition is admittedly mild, as it is satisfied for $g(q)=q^{\gamma}$ and all $\gamma>1$ and even for exponential costs $g(q)=\exp(q)-1$.

\textbf{Example:} Suppose $g(q)=q^2/2.$ Then, the buyer's utility from the optimal symmetric mechanism with $k$ firm is
\[U(k)=\frac{1}{2}\frac{k}{2k-1}-\frac{\alpha}{k+1},\]
which is a quasi-convex function. Thus, there exists an $\alpha_0$ such that for \[k^*(\alpha)=
\begin{cases}
1, & \alpha<\alpha_0;\\
n, & \alpha>\alpha_0.
\end{cases}\]

\cref{1orall} also explains the quasi-convex behavior of utility depicted in Figure 2 in \cite{gershkov2021theory} for a related additively separable setting.

\section{Conclusion}
\label{sec:conclude}
In this paper, we considered the problem of optimal procurement in the presence of both winner-pay (production) and all-pay (investment) costs, and contractible quality. For this, we built a special mechanism design framework in which outcomes can be conditioned on both messages and qualities. In this framework, we first characterized optimal symmetric mechanisms. We show that  an optimal symmetric mechanism is a scoring auction with a scoring rule that is either flatter or steeper than the well-known rule obtained under assumption of winner-pay costs only \citep{che1993design}. Crucially, even with ex ante symmetric firms, an optimal mechanism may be asymmetric, i.e., exhibit favoritism. This happens when the importance of private information is relatively low or the elasticity of investment costs with respect to quality is sufficiently high.

We characterized optimal mechanisms without the symmetry restriction which can be either sole sourcing, the symmetric scoring auction, or, depending on the curvature of marginal investment costs, a certain combination of the two. We identify two combinations that may be optimal. One involves a discriminating treatment of relatively inefficient types of two firms (a score floor auction) while the other discriminates between relatively efficient types of two firms (a score ceiling auction). Both involve side-payments that are needed to ensure the correct incentives of the favored bidder. 

It is hard to ignore that our mechanisms resemble corruption. Indeed, in the score ceiling auction, the side payment is paid to the auctioneer whenever the favored firm is guaranteed to win. Moreover, even our environment resembles those that are commonly associated with corruption: large investment costs and vanishing private information. But they are, in fact, quite different. Favoritism is efficient here, because it reduces the duplication of investment costs, and the auctioneer can not be bribed, as he is also the buyer. Thus, favoritism is not necessarily a sign of corruption. 

To conclude, it is well known that, in a procurement auction, buyer sometimes tries to give an advantage to his preferred bidder, not because he is corrupt, but because he believes that this bidder would do the best job. One can speculate that introducing asymmetric mechanisms into the procurement code may be a way to legitimize such favoritism.

\bibliography{scoring}
 
\newpage  
\part*{Appendix}
\appendix
 
\setcounter{page}{1} 

\section{Proofs for \cref{sect:symmetric}}
\label{app:symmetric}

Denote the indirect investment costs function $C(x,\theta)$ as $$C(x,\theta):=\min_q\{\tilde{C}^I(q,\theta):V(q)-\tilde{C}^P(q,\theta)=x\},$$ 
for all $x$ in the range of $V(q)-\tilde{C}^P(q, \theta)$.

\begin{lemm}\label{C-properties}
Under \cref{cost_properties,regularity}, the indirect investment costs function $C(x,\theta)$ is strictly increasing in both arguments, convex in $x$ and is supermodular in $(x,\theta)$.
\end{lemm}

\begin{proof}
Recall that $V$ is concave and $\tilde C^P$ is strictly convex in $q$, thus the equation $V(q)-\tilde{C}^P(q,\theta)=x$ has exactly two roots for all levels of $x$ except the highest feasible one. 

Since $\tilde C^I$ is increasing in $q$, it is optimal to pick the smaller root, which we denote as $\hat{q}(x,\theta)$. Clearly, $\hat{q}(x,\theta)$ is increasing in $x$, and so is the associated level of $\tilde C^I$. Thus, $C(x,\theta)$ is increasing in $x$.

Furthermore, since both $\tilde C^P, \tilde C^I$ are  increasing in $\theta$, $\hat{q}(x,\theta)$ is  increasing in $\theta$ and so is the associated level of $\tilde C^I$. Thus, $C(x,\theta)$ is increasing in $\theta$.

To obtain $C''_{xx}$ and $C''_{x\theta}$ we first compute $$\left.C'_x(x, \theta) = \frac{\tilde C^I_q(q,\theta)}{V'(q)-\tilde C^P_q(q,\theta)}\right|_{q = \hat q(x, \theta)}$$
which is the value of the Lagrange multiplier in the constrained optimization. Further differentiating with respect to $x$ or $\theta$ we obtain
\begin{gather*}\left.C''_{xx}(x, \theta) = \frac{\partial}{\partial q}\left(\frac{\tilde C^I_{q}(q,\theta)}{V'(q)-\tilde C^P_{q}(q,\theta)}\right)\right|_{q = \hat q(x, \theta)} \cdot \hat q'_x(x,\theta) \\
\left.C''_{x\theta}(x, \theta) = \frac{\partial}{\partial q}\left(\frac{\tilde C^I_{q}(q,\theta)}{V'(q)-\tilde C^P_{q}(q,\theta)}\right)\right|_{q = \hat q(x, \theta)} \cdot \hat q'_{\theta}(x,\theta) + \frac{\partial}{\partial \theta}\left.\left(\frac{\tilde C^I_{q}(q,\theta)}{V'(q)-\tilde C^P_{q}(q,\theta)}\right)\right|_{q = \hat q(x, \theta)}\end{gather*} which are both positive, since $\hat q$ is increasing in $x$ and $\theta$; $\tilde C^I_q, \tilde C^P_q$ are both increasing in $q$ and $\theta$; and $V'$ is decreasing in $q$. Thus, $C(x,\theta)$ is convex in $x$ and supermodular.
\end{proof}

The principal's utility may be written simply as
\begin{equation}\label{utility_x_app}
U=\mathbb{E}(\max_i x(\theta_i)-\sum_{i=1}^n C(x(\theta_i),\theta_i)). 
\end{equation}

\begin{lemm}\label{x-decreasing}
In an optimal symmetric mechanism, the virtual production surplus $x(q(\theta),\theta)$ is decreasing in $\theta$.    
\end{lemm}

\begin{proof}
Consider the reduced ($x$-only) problem of maximizing \eqref{utility_x_app} over the schedule $x(\theta_i):=x(q(\theta),\theta)$, which is the same for all bidders. That is, consider the problem
\begin{equation}\label{reduced_max}
\max\limits_{x(\theta)}\mbox{ }\mathbb{E}\max_i x(\theta_i)-n\mathbb{E} C(x(\theta),\theta),  
\end{equation}
where $C(x,\theta)$ is the indirect investment cost function and where we used symmetry to write the costs term. After $x(\theta)$ is chosen, the quality $q(\theta)$ is chosen to solve $\min_q \tilde{C}^I(q,\theta)$ $\mbox{s.t. }v(q)-\tilde{C}^P(q,\theta)=x(\theta)$. 

We show that for an arbitrary function $x(\theta)$ there exists a decreasing function $\tilde{x}(\theta)$ such that the objective \eqref{reduced_max} is larger under $\tilde{x}(\theta)$ than under $x(\theta)$. For that, we leverage that fact that by \cref{C-properties}, the indirect investment cost function $C(x,\theta)$ is supermodular. 

Namely, we take $\tilde{x}(\theta)$ to be the equimeasurable decreasing rearrangement of $x(\theta)$, that is, the decreasing function  $\tilde{x}(\theta)$ such that the probability distribution of $\tilde{x}(\theta)$ is the same as that of $x(\theta)$. If $G(x)$ is the cdf of $x(\theta)$, it is not hard to see that $\tilde{x}(\theta)$ satisfies $G(\tilde{x}(\theta))\equiv 1-F(\theta)$.

By construction, $\mathbb{E}\max_i x(\theta_i)$ is unchanged when $x(\theta)$ is replaced with $\tilde{x}(\theta)$. The second term, the negative costs, however, weakly increases by Theorem~3 from \cite{crowe1986rearrangements}, a general result about how an integral of a supermodular function of two inner functions changes when the inner functions are replaced by their rearrangements. That is, Theorem~3 from \cite{crowe1986rearrangements} implies that 
$\mathbb{E} C(\tilde{x}(\theta),\theta))\leq \mathbb{E} C(x(\theta),\theta))$ if $C(x,\theta)$ is supermodular, which it indeed is. Summing up, we get that the objective \label{reduced_max} weakly increases when $x(\theta)$ is replaced with $\tilde{x}(\theta)$; restricting attention to decreasing functions $x(\theta)$ is without loss. 
\end{proof}

\textbf{Remark:} Why did we need to employ rearrangements and not just standard monotone comparative statics results? The answer is that the monotonicity analysis is complicated by the symmetry constraint. Without it, we would have the problem 
\begin{equation*}
\max\limits_{x_i(\theta_i)}\mbox{ }\mathbb{E}\max_i x_i(\theta_i)-\sum_i C(x_i(\theta_i),\theta_i).
\end{equation*}
Then, fixing $x_j(\theta_j)$, the problem of optimization over $x_i(\theta_i)$ could be reduced to pointwise maximization to which we could apply standard monotone comparative statics. With symmetry and no a priori knowledge that $x(\theta)$ is decreasing, no reduction to pointwise optimization can be done.

\begin{lemm}\label{qstarproperties}
Let $q^*(z,\theta)$ be the solution to the problem
\[\max\limits_q\mbox{ }[s(q)-C^P(q,\theta)]z-C^I(q,\theta).\]
Then, $q^*(z,\theta)$ is increasing in $z$ and decreasing in $\theta$.    
\end{lemm}
\begin{proof}
By the first-order condition, $q^*(z,\theta)$ satisfies $[s'(q)-C^P_q(q,\theta)]z=C^I_q(q,\theta)>0$, so $s'(q)-C^P_q(q,\theta)>0$ in a neighborhood of $q^*(z,\theta)$. Thus, the objective function is supermodular in $(q,z)$ in a neighborhood of $q^*(z,\theta)$, which implies that $q^*(z,\theta)$ is increasing in $z$.

The fact that $q^*(z,\theta)$ is decreasing in $\theta$ follows simply from the fact that the objective is submodular in $(q,\theta)$ by our assumptions on costs (assumption~\ref{cost_properties}).
\end{proof}

\begin{lemm}\label{Sdecreasing}
For any $\theta_0\leq 1$, the score bidding strategy
\begin{equation}\label{Sdef}
S^*(\theta|\theta_0)=s(q(\theta))-C^P(q(\theta),\theta)-\frac{C^I(q(\theta),\theta)+\textit{IR}(\theta,\theta_0)}{(1-F(\theta))^{n-1}},  
\end{equation}
where $s(q)$ is given by \eqref{scoredef}, is decreasing in $\theta$ for $\theta\leq \theta_0$.
\end{lemm}
\begin{proof}
By straightforward differentiation, after simplifications using \eqref{scoredef}, we get 
\begin{equation}\label{Sprime}
\frac{\partial S^*(\theta|\theta_0)}{\partial\theta}=
-\frac{(n-1)f(\theta)}{(1-F(\theta))^n}\left(C^I(q(\theta),\theta)+\textit{IR}(\theta,\theta_0)\right)<0,    
\end{equation}
as $C^I>0$, $\textit{IR}(\theta,\theta_0)\geq 0$ for $\theta\leq \theta_0$.
\end{proof}
?
\noindent
\textbf{Proof of \cref{symmimplement}}

\begin{proof}
Choosing quality $q$ and price $p$ is equivalent to choosing quality $q$ and score $S=s(q)-p$. Given score strategy $S(\theta)$, price strategy is recovered by $p(\theta)=s(q(\theta))-S(\theta)$. So it is sufficient to establish the existence of equilibrium in quality-score pairs.

We shall show that such a scoring auction has a BNE $(q,S)=(q(\theta),S^*(\theta))$ where $S^*(\theta)$ is given by \eqref{eqstrat1} with $\theta_0=1$. That is, with a slight abuse of notation we write $S^*(\theta)$ for $S^*(\theta|1)$.

This amounts to showing the following \emph{triple} continuum of inequalities:
\begin{align}
\nonumber
\forall (\theta,q',s') \mbox{   }&
(s(q(\theta))-C^P(q(\theta),\theta)-S^*(\theta))\mathbb{P}(S^*(\theta)>\max_{j\neq i}S^*(\theta_j))-C^I(q(\theta),\theta)\geq\\ 
\label{equil_ineq}
& (s(q')-C^P(q',\theta)-s')\mathbb{P}(s'>\max_{j\neq i}S^*(\theta_j))-C^I(q',\theta)
\end{align}
The difficulty in showing \eqref{equil_ineq} is that every double deviation $(q',s')$ is feasible. To deal with this problem, we employ a sequential optimization approach. Namely, we imagine that a firm first chooses a score to bid and then chooses a quality with the score bid fixed. At the first (score-choice) stage a firm anticipates its optimal quality choice at the second stage. 

Denote by $\Pi(q,S|\theta_i)$ the profit of firm $i$ if every other firm plays according to the alleged BNE and firm $i$ chooses quality $q$ and score bid $S$. Define $\Pi_{\max}(S|\theta_i):=\max_q \Pi(q,S)|\theta_i)$. That is, $\Pi_{\max}(S|\theta_i)$ is the maximum profit firm $i$ can get by bidding score $S$ when the quality is adjusted optimally to the score bid. 

It is sufficient to show that (i) for every $\theta_i$, $\Pi_{\max}(S|\theta_i)$ is maximized by the score bid \eqref{eqstrat1} and that (ii) after choosing the score bid \eqref{eqstrat1} in the first stage, a firm will choose the given quality $q(\theta)$ at the second stage.  

We will start from the end, that is, with part (ii). 
By lemma~\ref{Sdecreasing}, $S^*(\theta)$ is strictly decreasing. Thus, when bidding the score $S^*(\theta)$ a firm wins with prob. $\mathbb{P}(S^*(\theta)>\max_{j\neq i}S^*(\theta_j))=(1-F(\theta))^{n-1}$. 

Having chosen the score bid $S=S^*(\theta)$, the firm solves
\[\max\limits_q \Pi(q,S^*(\theta)=\max_q\mbox{ }[s(q)-C^P(q,\theta)-S^*(\theta)](1-F(\theta))^{n-1}-C^I(q,\theta).\]
We shall show that $q=q(\theta)$ solves this problem, or, equivalently, that $q(\theta)=q^*((1-F(\theta))^{n-1},\theta)$ where recall that $q^*(z,\theta)$ is the solution to the problem
\[\max\limits_q\mbox{ }[s(q)-C^P(q,\theta)]z-C^I(q,\theta).\]

The first derivative of the objective is
\[\Pi_q(q,S^*(\theta))=[s'(q)-C^P_q(q,\theta)](1-F(\theta))^{n-1}-C^I_q(q,\theta).\]
After plugging $s'(q)$ from \eqref{scoredef}, it becomes
\[\Pi_q=[C_q^P(q,\theta(q))-C_q^P(q,\theta)](1-F(\theta))^{n-1}+\left(\frac{1-F(\theta)}{1-F(\theta(q))}\right)^{n-1}\cdot C^I_q(q,\theta(q))-C^I_q(q,\theta).\]
Now set $q=q(t)$, where $q(\cdot)$ is the quality function to implement. We get
\[\Pi_q=[C_q^P(q(t),t)-C_q^P(q(t),\theta)](1-F(\theta))^{n-1}+\left(\frac{1-F(\theta)}{1-F(t)}\right)^{n-1}\cdot C^I_q(q(t),t)-C^I_q(q(t),\theta).\]
By supermodularity of $C^P(q,\theta)$, $C^I(q,\theta)$, and the fact that $\left(\frac{1-F(\theta)}{1-F(t)}\right)^{n-1}<1$ iff $t<\theta$, we get that $\Pi_q<0$ if $t<\theta$, $\Pi_q>0$ if $t>\theta$ and $\Pi_q=0$ if $t=\theta$. As $q(\theta)$ is strictly decreasing, this implies that $\Pi_q>0$ if $q(t)<q(\theta)$, $\Pi_q<0$ if $q(t)>q(\theta)$, $\Pi_q=0$ if $q(t)=q(\theta)$. Thus, $\Pi(q,S^*(\theta))$ is strictly increasing for $q<q(\theta)$, and strictly decreasing for $q>q(\theta)$, which means that $q=q(\theta)$ is actually the optimal quality for the firm, as desired.

Now we turn to part (i). Clearly, it is suboptimal for a firm to bid any score bid outside of the range of $S^*(\theta)$ (given that the competitors are bidding according to $S^*(\theta)$). So we set $S=S^*(\tau)$ where $\tau$ is the type that the firm $i$ mimics when bidding $S=S^*(\tau)$. Abusing notation, we write $\Pi_{\max}(\tau|\theta)$ for $\Pi_{\max}(S^*(\tau)|\theta)$. 

So, 
\begin{equation}\label{profit_max}
  \Pi_{\max}(\tau|\theta)\equiv \max_q\mbox{ }(s(q)-C^P(q,\theta)-S^*(\tau))(1-F(\tau))^{n-1}-C^I(q,\theta). 
\end{equation}
Our task is to show that $\Pi_{\max}(\tau|\theta)$ is maximized at $\tau=\theta$. To this end, we will use Envelope Theorem to compute $\frac{\partial}{\partial \tau}\Pi_{\max}(\tau|\theta)$ and then show that $\frac{\partial}{\partial \tau}\Pi_{\max}(\tau|\theta)$ is positive for $\tau<\theta$ and negative for $\tau>\theta$ which implies global optimality of $\tau=\theta$.

First, we compute $\frac{\partial}{\partial \tau}[S^*(\tau)(1-F(\tau))^{n-1}]$. 
\begin{align}\nonumber
 & \frac{\partial}{\partial \tau}[S^*(\tau)(1-F(\tau))^{n-1}]= \\
 \nonumber
& \frac{\partial}{\partial \tau}[S^*(\tau)] (1-F(\tau))^{n-1}-(n-1)f(\tau)(1-F(\tau))^{n-2}S^*(\tau)=\\
\nonumber
 & -\frac{(n-1)f(\tau)}{(1-F(\tau))^n}\left(C^I(q(\tau),\tau)+\textit{IR}(\tau,1)\right)(1-F(\tau))^{n-1}-(n-1)f(\tau)(1-F(\tau))^{n-2}S^*(\tau))=\\
 \label{deriv}
 &-(n-1)(1-F(\tau))^{n-2}f(\tau)[s(q(\tau))-C^P(q(\tau),\tau)], 
\end{align}
where we used \eqref{Sprime} for $\frac{\partial}{\partial \tau}[S^*(\tau)]$ and \eqref{Sdef} for $S^*(\tau)$.

The optimal quality in \eqref{profit_max} is $q^*((1-F(\tau))^{n-1},\theta)$ where $q^*(z,\theta)$ is notation from the above.
Differentiating \eqref{profit_max}, using Envelope Theorem and then using \eqref{deriv}, we get
\begin{multline*}
\frac{\partial}{\partial \tau}\Pi_{\max}(\tau|\theta)=\\-(n-1)(1-F(\tau))^{n-2}f(\tau)(s(q^*((1-F(\tau))^{n-1},\theta))-C^P(q^*((1-F(\tau))^{n-1},\theta),\theta))-\frac{\partial}{\partial \tau}[S^*(\tau)(1-F(\tau))^{n-1}]=\\
(n-1)(1-F(\tau))^{n-2}f(\tau)\left[s(q(\tau))-s(q^*((1-F(\tau))^{n-1},\theta))+C^P(q^*((1-F(\tau))^{n-1},\theta),\theta)-C^P(q(\tau),\tau)\right].
\end{multline*}
In the proof of part (ii) above we showed that $q(\tau)\equiv q^*((1-F(\tau))^{n-1},\tau)$. For brevity, denote $\tilde{q}_{\tau}(u):=q^*((1-F(\tau))^{n-1},u)$. We showed that $\tilde{q}_{\tau}(\tau)\equiv q(\tau)$.
Also, because $q^*(z,\theta)$ is strictly increasing in $z$ by \cref{qstarproperties}, $\tilde{q}_{\tau}(u)< \tilde{q}_{u}(u)\equiv q(u)$ for $u< \tau$. 
It follows from the above that $\frac{\partial}{\partial \tau}\Pi_{\max}(\tau|\theta)$ has the same sign as
\begin{multline*}
s(\tilde{q}_{\tau}(\tau))-s(\tilde{q}_{\tau},\theta)+
C^P(\tilde{q}_{\tau}(\theta),\theta)-C^P(\tilde{q}_{\tau}(\tau),\tau)=\\
\int_{\theta}^{\tau} s'(\tilde{q}_{\tau}(u))d\tilde{q}_{\tau}(u)-\int_{\theta}^{\tau}C^P_q(\tilde{q}_{\tau}(u),u)d\tilde{q}_{\tau}(u)-
\int_{\theta}^{\tau}C^P_{\theta}(\tilde{q}_{\tau}(u),u)du.
\end{multline*}
Plugging $s'(q)$ from \eqref{scoredef} and rearranging, we finally get that
$\frac{\partial}{\partial \tau}\Pi_{\max}(\tau|\theta)$ has the same sign~as
\begin{align} \label{1}
 & \int_{\theta}^{\tau}[C^P_q\tilde{q}_{\tau}(u),\theta(\tilde{q}_{\tau}(u)))-C^P_q(\tilde{q}_{\tau}(u),u)]d\tilde{q}_{\tau}(u)+\\ \label{2} &\int_{\theta}^{\tau}\frac{C^I_q(\tilde{q}_{\tau}(u),\theta(\tilde{q}_{\tau}(u)))}{(1-F(\theta(\tilde{q}_{\tau}(u))))^{n-1}}d\tilde{q}_{\tau}(u)-\\ \label{3}
& \int_{\theta}^{\tau} C^P_{\theta}(\tilde{q}_{\tau}(u),u)du.
\end{align}
Suppose $\theta<\tau$. We discussed above that for $u<\tau$ $\tilde{q}_{\tau}(u)< q(u)$. Applying the decreasing function $\theta(q)$ to both sides, we get $\theta(\tilde{q}_{\tau}(u))>\theta(q(u))\equiv u$. 
Hence, due to supermodularity of $C^P(q,\theta)$, $C^P_q\tilde{q}_{\tau}(u),\theta(\tilde{q}_{\tau}(u)))-C^P_q(\tilde{q}_{\tau}(u),u)>0$. Adding the fact that by lemma~\ref{qstarproperties}
$d\tilde{q}_{\tau}(u)<0$, we conclude that for $\theta<\tau$ the term \eqref{1} is negative. Likewise, for $\theta<\tau$ the term \eqref{2} is negative, as $C_q^I>0$ and $d\tilde{q}_{\tau}(u)<0$. The term \eqref{3} is positive for $\theta<\tau$ as $C^P_{\theta}>0$. To sum up, we get that for $\tau>\theta$ $\frac{\partial}{\partial \tau}\Pi_{\max}(\tau|\theta)<0$. Analogously, for $\tau<\theta$ $\frac{\partial}{\partial \tau}\Pi_{\max}(\tau|\theta)>0$; obviously, for 
$\tau=\theta$ $\frac{\partial}{\partial \tau}\Pi_{\max}(\tau|\theta)=0.$

We conclude that $\Pi_{\max}(\tau|\theta)$ is strictly increasing for $\tau<\theta$ and strictly decreasing for $\tau>\theta$; $\tau=\theta$ is the global maximum of $\Pi_{\max}(\tau|\theta)$, as desired. Together with part (2) shown above this implies that $(q(\theta), S^*(\theta))$ is indeed every firm's strategy in a symmetric BNE. 
\end{proof}

In proof above we showed not only that it is \emph{optimal} to bid $S^*(\tau)=S^*(\theta)$ but that the objective function $\Pi_{\max}(\tau|\theta)$ is \emph{single-peaked} (strictly quasi-concave) in $\tau$. This a stronger property that will be extremely useful in proving the correctness of implementation in the asymmetric case. For future reference, we state in the following lemma:
\begin{lemm}\label{single-peaked}
The score-bid-only objective function $\Pi_{\max}(\tau|\theta)\equiv \max_q\Pi(q,S^*(\tau)|\theta)$ is single-peaked in $\tau$ with the peak at $\tau=\theta$.
\end{lemm}
\begin{proof}
Follows from the last part of proof of \cref{symmimplement}.
\end{proof}

\noindent
\textbf{Proof of \cref{elmatters}}
\begin{proof}
From \eqref{scoredef}, we have
\begin{equation}\label{optscore_ce}
s'^{*}(q)\equiv \theta^{E_P }(q)g'_1(q)+\frac{\beta\theta^{E_I }(q)g'_2(q)}{(1-F(\theta(q)))^{n-1}},   
\end{equation}
where $\theta(q)$ is the inverse of the optimal $q^*(\theta,n,\beta)$. We suppress the dependence of $\theta(q)$ on $\beta$ and $n$ for brevity.

At the same time, from the optimality of $\theta(q)$,
\begin{equation}\label{optqual_ce}
(1-F(\theta(q)))^{n-1}(v'(q)-(1+E_P \delta)\theta^{E_P }(q)g_1'(q))-\beta(1+E_I \delta)\theta^{E_I }(q)g_2'(q)\equiv 0.   
\end{equation}
Solving \eqref{optqual_ce} for $\beta\theta^{E_I }(q)g_2'(q)$ and plugging this in \eqref{optscore_ce}, we get
\[s'^*(q)\equiv \theta^{E_P }(q)g_1'(q)+\frac{v'(q)-(1+E_P \delta)\theta^{E_P }(q)g_1'(q)}{1+E_I \delta}=\frac{v'(q)+\delta(E_I -E_P )\theta^{E_P }(q)g'_1(q)}{1+E_I \delta}.\]
It follows from \eqref{symm_obj} that $\theta(q)$ is decreasing in both $\beta$ and $n$ for a fixed $q$. From this, the result is immediate.
\end{proof}

\section{Proofs for \cref{sec:suboptimal}}
\label{app:suboptimal}

\noindent
\textbf{Proof of \cref{smallalpharesult}}

\begin{proof}
We will show the suboptimality of the optimal symmetric mechanism for $\alpha=0$. Namely, we will show that for $\alpha=0$ the completely asymmetric mechanism in which only one bidder is left dominates the optimal symmetric mechanism. Then the result follows from continuity (the objective at the completely asymmetric mechanism and the objective at the optimal symmetric mechanism are continuous in $\alpha$).

When $\alpha=0$, the indirect investment costs $C(x,\theta)$ is just a function of $x$. Abusing notation, we write $C(x)$.
Denote by $H_i(x)$ the cdf of $x_i(\theta_i)$. We can write the objective as a functional of $H_i(x)$ and optimize over $H_i(x)$. 
That is,
\begin{align}\label{H-objective}
\nonumber
U & = \mathbb{E}\max_i x_i(\theta_i)-\mathbb{E}\sum_{i=1}^n C(x_i(\theta_i))\\
\nonumber
& =\int_0^{+\infty}\left(1-\prod_{i=1}^nH_i(x)\right)dx-\int_0^{+\infty}\sum_{i=1}^nC'(x)(1-H_i(x))dx\\
& =\int_0^{+\infty}\left(1-\prod_{i=1}^nH_i(x)-\sum_{i=1}^nC'(x)(1-H_i(x))\right)dx.
\end{align}
Ignoring the monotonicity constraint for $H_i(x)$, let's optimize \eqref{H-objective} pointwise. The integrand is linear in each of $H_i$, so without loss of generality at every $x$, $H_i(x)\in\{0,1\}$. It is easy to see that at optimum at every $x$ at most one of $H_i(x)$ is 0; otherwise the objective can be increased since $C'(x)\geq 0$ for $x>0$. It remains to compare two cases: all of $H_i$ are 1 and all but one of $H_i$ is 1, and the remaining one is zero. In the first case, the integrand is 0; in the second case, it is $1-C'(x)$. Thus, when $0\geq 1-C'(x)$, it is optimal to set all $H_i(x)$ to 1; if  $0<1-C'(x)$, it is optimal to set all but one $H_i(x)$ to 1. 
Suppose, without loss of generality, that bidder 1 is the unique bidder with $H_i(x)=0$ for $C'(x)<1$. Then, the optimal cds of $x_i(\theta_i)$ are
\begin{equation}\label{optmech}
H^*_1(x)=
\begin{cases}
0, & C'(x)<1;\\
1, & C'(x)\geq 1;
\end{cases}\quad\mbox{and for }i>1,  
H^*_i(x)=1\mbox{ for all }x\geq 0.
\end{equation}
These cdfs of $x_i(\theta_i)$ are readily implementable. 

The functions $H^*_i$ are nondecreasing, hence this is indeed the optimal mechanism.
This is the completely asymmetric mechanism in which only one bidder is allowed to participate and this bidder produces the monopoly surplus $x_m$ that satisfies $C'(x_m)=1$.

The optimal symmetric mechanism $M^*_{sym}$ maximizes \eqref{H-objective} pointwise under the constraint that $H_i(x)=H_j(x)$ for all $i,j$. Taking FOC, one gets $H^*_{sym}(x)=\sqrt[n-1]{C'(x)}$. After changing $H_1(x)$ from $H^*_{sym}(x)$ to $H_1^*(x)$ given by \eqref{optmech} the objective does not change but if after that $H_2(x)$ is changed from $H^*_{sym}(x)$ to $H_1^*(x)$ given by \eqref{optmech}, the objective strictly increases; hence $M^*_{sym}$ is not optimal and is strictly dominated by $M^*$.

Write $M^*(\alpha)$ and $M^*_{sym}(\alpha)$ for the optimal mechanism and the optimal symmetric mechanism as functions of $\alpha$. Write $\mathbb{E}U(M,\alpha)$ as the principal's expected utility as the function of mechanism and $\alpha$. We have shown above that 
\[\mathbb{E}U(M^*(0),0)>\mathbb{E}U(M_{sym}^*(0),0).\]
By continuity, this implies that
\[\mathbb{E}U(M^*(0),\alpha)>\mathbb{E}U(M_{sym}^*(\alpha),\alpha)\]
for all sufficiently small $\alpha$.\eop
\end{proof}

\noindent
\textbf{Proof of \cref{symmetricbad2}}
\begin{proof}
We will dominate the optimal symmetric mechanism with a ``score floor auction'' which is discussed in detail in section~\ref{sect:asym}. Denote by $x^*_{sym}(\theta)$ the optimal symmetric  $x$-outcome function. The score floor mechanism induces the $x$ outcomes as in \eqref{x-example} and similar to its quality outcomes \eqref{q_floor}:
\begin{gather}
x^*_1(\theta_1) = x^*_{sym}(\min\{\theta_0,\theta_1\}), \quad 
x^*_2(\theta_2) = \begin{cases}
x^*_{sym}(\theta_2), & \theta_2 < \theta_0\\
x^*_{sym}(1), & \theta_2 > \theta_0,
\end{cases}
\end{gather}
where $\theta_0$ is the threshold creating asymmetry. Note that if $\theta_0$ these outcomes correspond to optimal symmetric outcomes.
We will show that these $x$-outcome functions yield a higher utility for the principal if the threshold $\theta_0$ is slightly lower than 1 (so the mechanism is slightly asymmetric).

Set $n=2$. The principal's utility as a function of threshold $\theta_0$ can be written as
\[U(\theta_0)=n\int_0^{\theta_0}\max\limits_x\left[(1-F(\theta))^{n-1}x-C(x,\theta)\right]d\theta+\int_{\theta_0}^1\max\limits_x\left[(1-F(\theta_0))^{n-1}x-C(x,\theta)\right]d\theta.\]
Denote by $x^*(\theta,\theta_0)$ the solution to the problem $\max\limits_x\left[(1-F(\theta_0))^{n-1}x-C(x,\theta)\right]$. Note that $x^*_{sym}(\theta)\equiv x^*(\theta,\theta)$.
Differentiating $U(\theta_0)$, we get 
\[
U'(\theta_0)/f(\theta_0)=(n-1)\max\limits_x\left[(1-F(\theta_0))^{n-1}x-C(x,\theta_0)\right]-(n-1)(1-\theta_0)^{n-2}\int_{\theta_0}^1 x^*(\theta,\theta_0)d\theta,   
\]
where we used Envelope theorem to evaluate the derivative of the second addend in $U(\theta_0)$. This may be further rewritten as
\begin{equation}\label{U'}
U'(\theta_0)/(n-1)f(\theta_0)=(1-F(\theta_0))^{n-2}\left((1-F(\theta_0))x^*(\theta_0,\theta_0)-\int_{\theta_0}^1 x^*(\theta,\theta_0)d\theta\right)-C(x^*(\theta_0,\theta_0),\theta_0).   
\end{equation}
In what follows, we prove that $\lim\limits_{\theta_0\to 1}\frac{U'(\theta_0)}{(1-F(\theta_0))^{\frac{\gamma(n-1)}{\gamma-1}}}<0$ which implies the result. It follows from \eqref{U'} that
\[\lim\limits_{\theta_0\to 1}\frac{U'(\theta_0)}{(1-F(\theta_0))^{\frac{\gamma(n-1)}{\gamma-1}}}=(n-1)(L_1-L_2),\]
where 
\[L_1:=\lim\limits_{\theta_0\to 1}\frac{(1-F(\theta_0))^{n-2}\left((1-F(\theta_0))x^*(\theta_0,\theta_0)-\int_{\theta_0}^1 x^*(\theta,\theta_0)d\theta\right)}{(1-F(\theta_0))^{\frac{\gamma(n-1)}{\gamma-1}}},\]
and
\[L_2:=\lim\limits_{\theta_0\to 1}\frac{C(x^*(\theta_0,\theta_0),\theta_0)}{(1-F(\theta_0))^{\frac{\gamma(n-1)}{\gamma-1}}}.\]
To evaluate the limit $L_1$, we
use the first-order Taylor expansion of the function $t\to x^*(t,\theta_0)$ at $t=1$ and the 
second-order Taylor expansion of the function $t\to \int_{t}^1 x^*(\theta,\theta_0)d\theta$ at $t=1$. (The expansions are valid because $x^*(\theta,\theta_0)$ is differentiable in the first argument by Implicit Function theorem.) After simplifications, we get
\begin{align*}
 L_1=& \lim\limits_{\theta_0\to 1}\frac{(1-F(\theta_0))^{n-2}\left((1-F(\theta_0))x^*_{\theta}(1,\theta_0)(\theta_0-1)+x^*_{\theta}(1,\theta_0)(1-F(\theta_0))^2/2+o((1-F(\theta_0))^2)\right)}{(1-F(\theta_0))^{\frac{\gamma(n-1)}{\gamma-1}}}=\\
 & \lim\limits_{\theta_0\to 1}\frac{-x^*_{\theta}(1,\theta_0)(1-F(\theta_0))^n/2+o((1-F(\theta_0))^n)}{(1-F(\theta_0))^{\frac{\gamma(n-1)}{\gamma-1}}}=0,
\end{align*}
because $x^*_{\theta}(1,1)$ is finite and $n>\frac{\gamma(n-1)}{\gamma-1}$ since $\gamma>n=2$.

Now consider the limit $L_2$. Obviously, $L_2\geq 0$. We need a stronger statement that $L_2>0$ which we prove below.

First, denoting $\lim\limits_{x\to\underline{x}^+}\frac{C(x,1)}{(x-\underline{x})^{\gamma}}:=D>0$, by l'Hospital's rule we get that $\lim\limits_{x\to\underline{x}^+}\frac{C_x(x,1)}{(x-\underline{x})^{\gamma-1}}=\gamma D>0$. Note that this implies that $C_x(\underline{x},1)=0$.
Thus, 
\[\lim\limits_{x\to\underline{x}^+}\frac{C(x,1)}{[C_x(x,1)]^{\frac{\gamma}{\gamma-1}}}=\left[\lim\limits_{x\to\underline{x}^+}\frac{[C(x,1)]^{\gamma-1}}{[C_x(x,1)]^{\gamma}}\right]^{\frac{1}{\gamma-1}}=\left[\lim\limits_{x\to\underline{x}^+}\frac{D^{\gamma-1}(x-\underline{x})^{\gamma(\gamma-1)}}{(\gamma D)^{\gamma}(x-\underline{x})^{(\gamma-1)\gamma}}\right]^{\frac{1}{\gamma-1}}=\frac{D}{(\gamma D)^{\frac{\gamma}{\gamma-1}}}>0.\]

Now substitute $x=x^*(1,\theta_0)$ with $\theta_0\to 1$ to the above limit (this is valid since $x^*(1,1)=\underline{x}$ which follows from $C_x(\underline{x},1)=0$.). Note that by FOC $ C_x(x^*(1,\theta_0),1)\equiv (1-F(\theta_0))^{n-1}$. Thus,
\[0<\frac{D}{(\gamma D)^{\frac{\gamma}{\gamma-1}}}=\lim\limits_{x\to\underline{x}^+}\frac{C(x,1)}{[C_x(x,1)]^{\frac{\gamma}{\gamma-1}}}=\lim\limits_{\theta_0\to 1}\frac{C(x^*(1,\theta_0),1)}{[C_x(x^*(1,\theta_0),1)]^{\frac{\gamma}{\gamma-1}}}=\lim\limits_{\theta_0\to 1}\frac{C(x^*(1,\theta_0),1)}{(1-F(\theta_0))^{\frac{\gamma(n-1)}{\gamma-1}}}.\]
Finally, note that \[0<\lim\limits_{\theta_0\to 1}\frac{C(x^*(1,\theta_0),1)}{(1-F(\theta_0))^{\frac{\gamma(n-1)}{\gamma-1}}}=\lim\limits_{\theta_0\to 1}\lim\limits_{\theta\to 1}\frac{C(x^*(\theta,\theta_0),\theta)}{(1-F(\theta_0))^{\frac{\gamma(n-1)}{\gamma-1}}}=\lim\limits_{\substack{\theta\to 1\\\theta_0\to 1}}\frac{C(x^*(\theta,\theta_0),\theta)}{(1-F(\theta_0))^{\frac{\gamma(n-1)}{\gamma-1}}}=L_2.\]

To sum up, 
\[\lim\limits_{\theta_0\to 1}\frac{U'(\theta_0)}{(1-F(\theta_0))^{\frac{\gamma(n-1)}{\gamma-1}}}=(n-1)(L_1-L_2)=(n-1)(0-L_2)<0.\]
\end{proof}

\section{Proofs for \cref{sect:asym}}
\label{app:asym}

\noindent
\textbf{Proof of \cref{main result}}:

\begin{proof}
Suppose marginal investment costs are sufficiently convex. By definition, if marginal costs are sufficiently convex the function $q\to \alpha \xi(g'(q)))-q=\alpha(1-J(F^{-1}(1-g'(q))))-q$ is strictly quasi-convex. Quasi-convexity is preserved under any monotone transformation of the argument, so plugging $q=g'^{-1}(1-F(\theta))$, we get that the function $\alpha(1- J(\theta))-g'^{-1}(1-F(\theta))$ is strictly quasi-convex as well.

In \cref{zhangzhang}, we discuss the relation to \cite{zhang2017auctions} at length and show the equivalence  between our models under \cref{separablestructure}.

By \cref{to_Zhang} in \cref{zhangzhang}, 
the optimal allocation $z^*_i(\bm\theta)$ and action schedules $a^*_i(\theta_i)$ in our problem under \cref{separablestructure} are the same as the optimal allocation $z^*_i(\bm\theta)$ and equilibrium action schedules $a^*_i(\theta_i)$ in the problem in \cite{zhang2017auctions} in which an agent's utility is $V(a)+\alpha\mu-\alpha$, investment costs are $g(a)$ and the distribution of types is $F^{\mu}(\mu)=1-F(1-\mu)$. (We use the notation $\mu$ for type in \cite{zhang2017auctions} to distinguish it from the type in our paper; the two are different because utility decreases in type in our paper while increases in type in \cite{zhang2017auctions}; to make the connection, we set $\theta=1-\mu$.)

It remains to employ the analysis in \cite{zhang2017auctions} to determine the structure of optimal allocation under sufficiently convex and sufficiently concave marginal costs. To this end, let's translate the function $\alpha(1- J(\theta))-g'^{-1}(1-F(\theta))$ into the language of \cite{zhang2017auctions}.
\begin{align*}
 \alpha(1- J(\theta))-g'^{-1}(1-F(\theta)) & = \alpha\left(1-\theta-\frac{F(\theta)}{f(\theta)}\right)-g'^{-1}(1-F(\theta)) \\
 & = \alpha\left(\mu-\frac{F(1-\mu)}{f(1-\mu)}\right)-g'^{-1}(1-F(1-\mu))\\
 & = \alpha\left(\mu-\frac{1-F^{\mu}(\mu)}{f^{\mu}(\mu)}\right)-g'^{-1}(F^{\mu}(\mu)),
\end{align*}
where we used the definitions of $\mu$ and $F^{\mu}(\cdot)$.

Because the function $\theta\to \alpha(1- J(\theta))-g'^{-1}(1-F(\theta))$ was strictly quasi-convex and the substitution $\mu=1-\theta$ is monotone, the function 
\begin{equation}\label{our_function}
\alpha\left(\mu-\frac{1-F^{\mu}(\mu)}{f^{\mu}(\mu)}\right)-g'^{-1}(F^{\mu}(\mu))    
\end{equation}
is strictly quasi-convex as well.

\cite{zhang2017auctions} considers only quadratic investment costs, $g(a)=a^2/2K$, and the agent's utility function $a+\mu$. She finds that the structure of the optimal allocation is determined by the monotonicity properties of the function 
\begin{equation}\label{zhang_function}
\mu-\frac{1-F^{\mu}(\mu)}{f^{\mu}(\mu)}-KF^{\mu}(\mu).     
\end{equation}

By following the same steps as in \cite{zhang2017auctions} but with more general investment costs $g(a)$ and the utility $a+\alpha\mu$, we reach the conclusion that the appropriate generalization of \eqref{zhang_function}
is exactly \eqref{our_function}.
(This can also be derived using the approach in \cite{gershkov2021theory}.) 
We will invoke the appropriate generalizations of the results in \cite{zhang2017auctions}, replacing \eqref{zhang_function} with \eqref{our_function}. 
As \eqref{our_function} is strictly quasi-convex, it either (i) strictly decreasing, (ii) strictly increasing, or (iii) first strictly decreasing and then strictly increasing. In the case (i) by Corollary 1, part I in \cite{zhang2017auctions}, the optimal mechanism is symmetric so the quality schedules \eqref{q_floor} with  $\theta_0=1$ ($\mu_0=0$) are optimal. In the case (ii) by Theorem 1, part I in \cite{zhang2017auctions}, the optimal mechanism is completely asymmetric so that one bidder gets the contract with probability 1. This corresponds to quality schedules \eqref{q_floor} with $\theta_0=0$ ($\mu_0=1$). Finally, in the case (iii) by Theorem 1 in \cite{zhang2017auctions} there exists 
$\mu_0\in[0,1]$ such that an optimal mechanism is symmetric for 
$\mu>\mu_0$ while having $[0,\mu_0]$ is a ``favored bidder interval''. That is, for $\mu<\mu_0$ one bidder is ``favored'' and gets the contract with probability $F^{\mu}(\mu_0)$ and the other bidder is ``unfavored'' and gets the contract with probability 0. Translating this back into world of $\theta$, we get the allocation functions \eqref{right_allocation} with $\theta_0=1-\mu_0$
The quality functions \eqref{q_floor} then follow either from principal's optimization in our model or from agent's optimization in \cite{zhang2017auctions}, as in lemma~\ref{to_Zhang}.

The case of sufficiently concave marginal investment costs is considered analogously.
\end{proof}

\noindent
\textbf{Proof of \cref{implement-1}:}
\begin{proof}
We shall show that the score and quality strategies specified in the proposition constitute a BNE in the score floor auction.

{We again split the analysis of the optimality of the conjectured strategy into (i) the choice of the score when quality is chosen optimally and (ii) the choice of quality when the score is chosen at the equilibrium level. 

Part (ii), when the score is above the floor, follows from the analysis of the symmetric case in \cref{symmimplement}. When the unfavored firm is playing the score floor, she is guaranteed to lose. Thus, she produces the lowest possible quality $q = 0$, effectively exiting. When the favored firm is playing the score floor, she faces the same probability of winning as the threshold type $\textit{PW}_{sym}(\theta_0)$ and receives the bonus, independently of her own type. Moreover, type and quality do not interact under \cref{separablestructure} when evaluating the costs. Thus, she produces the same quality as the threshold type, $q_{sym}^*(\theta_0)$.}

To prove part (i), suppose a firm $i$ bids the score $S^*(\tau|\theta_0)$ while the alleged BNE strategy is $S_i^*(\theta|\theta_0)$. We shall show that it is optimal for the firm to pick $\tau=\theta$.

Let $\Pi_{\max}^{sym}(\tau|\theta)$ by the profit function in the symmetric mechanism considered in the proof of \cref{symmimplement} (the maximal-over-quality profit if the score bid is $S^*(\tau)$).  Let $\Pi_{\max}^{floor,1}(\tau|\theta,\theta_0)$ be the analogous profit function of the favored bidder in the score-floor asymmetric mechanism with threshold $\theta_0$. Let $\Pi_{\max}^{floor,2}(\tau|\theta,\theta_0)$ be the profit function of the unfavored bidder in this mechanism.

Recall that $\textit{IR}(a,b):=\int_{a}^b[(1-F(u))C^P_{\theta}(q^*(1-F(u),u),u)+C^I_{\theta}(q^*(1-F(u),u),u)]du$. It is not hard to see that, given the equilibrium score strategies, $\Pi_{\max}^{sym}(\tau|\theta)$ and $\Pi_{\max}^{floor,i}(\tau|\theta,\theta_0)$ are related for $\tau\leq \theta_0$ by
\begin{gather}\label{profits-relation_fav} 
\Pi_{\max}^{floor,1}(\tau|\theta,\theta_0)= \Pi_{\max}^{sym}(\tau|\theta)-\textit{IR}(\theta_0,1)+B\\
\label{profits-relation} 
\Pi_{\max}^{floor,2}(\tau|\theta,\theta_0)= \Pi_{\max}^{sym}(\tau|\theta)-\textit{IR}(\theta_0,1),
\end{gather}
where $B$ is the bonus for the favored bidder.

We consider two cases, with two sub-cases each.

\textbf{Case 1a. Favored bidder, true type $\theta<\theta_0$.}
Bidding a score $S\geq S_r$ is the same as bidding a score $S^*(\tau)$ for some $\tau\leq \theta_0$. By \eqref{profits-relation_fav}, $\Pi_{\max}^{floor,1}(\tau|\theta,\theta_0)$ differs from $\Pi_{\max}^{sym}(\tau|\theta)$ by a constant, hence it is maximal at the same value of $\tau$ which, as shown in the proof of \cref{symmimplement}, is $\tau=\theta$. Thus, if the bidder enters, it is optimal to bid $S^*(\theta)$, as desired. And the bidder will enter, as her profit when choosing $\tau=\theta$ is
\begin{align*}
 \Pi_{\max}^{sym}(\theta|\theta)-\textit{IR}(\theta_0,1) +B =
\textit{IR}(\theta,1) -\textit{IR}(\theta_0,1) +B =
\textit{IR}(\theta,\theta_0) +B >0,   
\end{align*}
where the first equality is from the proof of \cref{symmimplement}, the second is by the definition of $\textit{IR}(a,b)$, and the inequality is by the definition of $\textit{IR}(a,b)$ and $B$.

\textbf{Case 1b. Favored bidder, true type $\theta>\theta_0$.}
We need to show that it is optimal for the favored bidder to bid the reserve score $S_r=S^*(\theta_0)$, that is, to choose $\tau=\theta_0$. $\Pi_{\max}^{floor,1}(\tau|\theta,\theta_0)$ differs from $\Pi_{\max}^{sym}(\tau|\theta)$ by a constant, but by \cref{single-peaked},  $\Pi_{\max}^{sym}(\tau|\theta)$ is for every $\theta$ single-peaked in $\tau$ with the peak at $\tau=\theta$. It follows that $\Pi_{\max}^{floor,1}(\tau|\theta,\theta_0)$ has the same property.
Thus, for $\tau<\theta_0<\theta$ we have 
\[\Pi_{\max}^{floor,1}(\tau|\theta,\theta_0)<\Pi_{\max}^{floor,1}(\theta_0|\theta,\theta_0)<\Pi_{\max}^{floor,1}(\theta|\theta,\theta_0).\]
The first inequality says that for the types $\theta>\theta_0$ it is optimal to bid the reserve score $S_r=S^*(\theta_0|\theta_0)$. It remains to check that the mechanism is individually rational for the favored bidder with a type $\theta>\theta_0$, with the type $\theta=1$ having zero profit. This is ensured by the correct value of bonus $B$. Indeed,
\begin{align*}
\Pi_{\max}^{floor,1}(\theta_0|\theta,\theta_0)= & \Pi_{\max}^{sym}(\theta_0|\theta)-\textit{IR}(\theta_0,1) +B\geq\\
& \Pi_{\max}^{sym}(\theta_0|1)-\textit{IR}(\theta_0,1) +B=\\
& \Pi_{\max}^{sym}(\theta_0|1)-\Pi_{\max}^{sym}(\theta_0|\theta_0)+B=\\
& -B+B=0,
\end{align*}
where the inequality follows from the fact that $\Pi_{\max}^{sym}(\tau|\theta)$ is decreasing in $\theta$, the second equality follows from \cref{symmimplement} and the third equality follows from the fact that 
$\Pi_{\max}^{sym}(\theta_0|1)-\Pi_{\max}^{sym}(\theta_0|\theta_0)=-\int_{\theta_0}^1[(1-F(\theta_0))C^P_{\theta}(q^*(1-F(\theta_0),u),u)+C^I_{\theta}(...)]du=-B$ for $B=\alpha(1-\theta_0)(1-F(\theta_0)$, $C^P_{\theta}=\alpha$ and $C^I_{\theta}=0$. For $\theta=1$, the inequality holds as equality, so the type $\theta=1$ gets a profit of 0.

\textbf{Case 2a. Unfavored bidder, true type $\theta<\theta_0$.}
By the proof of \cref{symmimplement}, $\tau=\theta$ maximizes $\Pi_{\max}^{sym}(\tau|\theta)$.
By \eqref{profits-relation}, for $\tau\leq \theta_0$ $\Pi_{\max}^{floor}(\tau|\theta,\theta_0)$ differs from  $\Pi_{\max}^{sym}(\tau|\theta)$ by a constant, thus $\tau=\theta$ maximizes $\Pi_{\max}^{floor}(\tau|\theta,\theta_0)$ over all $\tau\leq\theta_0$. The profit from bidding $\tau=\theta$ is, by \eqref{profits-relation},
\begin{align*}
\Pi_{\max}^{floor,2}(\tau|\theta,\theta_0)= \Pi_{\max}^{sym}(\theta|\theta)-\textit{IR}(\theta_0,1) =
& \textit{IR}(\theta,1) -\textit{IR}(\theta_0,1) =\textit{IR}(\theta,\theta_0)>0,   
\end{align*}
On the other hand, bidding a score $S<S_r=S^*(\theta_0|\theta_0)$ would yield the bidder a maximal profit of 0 and so is no better than bidding $S^*(\theta|\theta_0)$. Thus, choosing $\tau=\theta$ is optimal and individually rational.

\textbf{Case 2b. Unfavored bidder, true type $\theta>\theta_0$.}
Choosing any $\tau\leq \theta_0$ would yield such a bidder a profit of
\begin{equation*}
\Pi_{\max}^{sym}(\tau|\theta)-\textit{IR}(\theta_0,1) \leq 
\Pi_{\max}^{sym}(\theta|\theta)-\textit{IR}(\theta_0,1) =
\textit{IR}(\theta,1)-\textit{IR}(\theta_0,1)=-\textit{IR}(\theta_0,\theta) < 0,
\end{equation*}
where the inequality follows from the fact that $\tau=\theta$ maximizes $\Pi_{\max}^{sym}(\tau|\theta)$ by \cref{symmimplement}.
{Thus, it is better to choose the score floor, effectively exiting the auction.} The unfavored bidder with type $\theta=1$ gets a profit of 0.
\end{proof}

\noindent
\textbf{Proof of \cref{implement-2}:}
\begin{proof}
We shall show that the score and quality strategies specified in the proposition constitute a BNE in score ceiling auction.

{We again split the analysis of the optimality of the conjectured strategy into (i) the choice of the score when quality is chosen optimally and (ii) the choice of quality when the score is chosen at the equilibrium level. 

Part (ii), when the score is below the ceiling, follows from the analysis of the symmetric case in \cref{symmimplement}. When the favored firm is playing the score ceiling, she is guaranteed to win. Thus, she produces the optimal monopolistic quality $q_{sym}^*(0)$. When the unfavored firm is playing the score ceiling, she faces the same probability of winning as the threshold type $\textit{PW}_{sym}(\theta_0)$ and pays a kickback, independently of her own type. Moreover, type and quality do not interact under \cref{separablestructure} when evaluating the costs. Thus, she produces the same quality as the threshold type, $q_{sym}^*(\theta_0)$.}

To prove part (i), suppose a firm $i$ bids the score $S^*(\tau)$ while the alleged BNE strategy is $S_i^*(\theta)$. We shall show that it is optimal for the firm to pick $\tau=\theta$.

Recall that  $\Pi_{\max}^{sym}(\tau|\theta)$ is the profit function in the symmetric mechanism considered in the proof of \cref{symmimplement} (the maximal-over-quality profit if the score bid is $S^*(\theta)$).  Let $\Pi_{\max}^{ceil,1}(\tau|\theta,\theta_0)$ be the analogous profit function of the favored bidder in the score-ceiling asymmetric mechanism with threshold $\theta_0$. Let $\Pi_{\max}^{ceil,2}(\tau|\theta,\theta_0)$ be the profit function of the unfavored bidder in this mechanism.

It is clear that it is never optimal for the favored bidder to bid a score $S>\bar{S}$ since bidding a slightly lower score gives her the same winning probability of 1, but ensures a higher price. At the same time, the unfavored bidder would not bid $S>\bar{S}$ since her bid won't count in that case. So both bidders won't bid scores higher than $\bar{S}=S^*(\theta_0)$. This means that we can restrict attention to deviations $\tau\geq \theta_0$.

We again consider two cases, with two sub-cases each.

For $i=1,2$ and all $\tau>\theta_0$
\begin{equation}\label{asym=sym}
\Pi_{\max}^{ceil,i}(\tau|\theta,\theta_0)=\Pi_{\max}^{sym}(\tau|\theta).
\end{equation}

\textbf{Case 1a. Favored bidder, true type $\theta>\theta_0$.}
By \eqref{asym=sym}, for $\tau>\theta_0$ $\Pi_{\max}^{ceil,1}(\tau|\theta,\theta_0)=\Pi_{\max}^{sym}(\tau|\theta)$, so unprofitability of such deviations follows from \cref{symmimplement}. 

Now consider the deviation $\tau=\theta_0$ whereby the favored bidder bids $\bar{S}=S^*(\theta_0)$ and wins for sure.
In this case, she will get a profit of
\[\Pi_{\max}^{ceil,1}(\theta_0|\theta,\theta_0)=\max_q(q-g(q))-\alpha\theta-\bar{S}-K.\]
By definition of the score ceiling, \[\bar{S} = S^*_{sym}(\theta_0) = \max_q\left[q-\frac{g(q)}{1-F(\theta_0)}\right]-\frac{\textit{IR}(\theta_0,1)}{1-F(\theta_0)},\]
while the kickback $K$ is by \eqref{kickback1}:
\begin{align*}
K=& \max_q[q-g(q)]-\max_q\left[q-\frac{g(q)}{1-F(\theta_0)}\right]+\alpha\frac{F(\theta_0)}{1-F(\theta_0)}\int_{\theta_0}^1(1-F(u))du\\
& = \max_q[q-g(q)]-\max_q\left[q-\frac{g(q)}{1-F(\theta_0)}\right]+\frac{F(\theta_0)}{1-F(\theta_0)}\textit{IR}(\theta_0,1).
\end{align*}
Taking stock, 
\[\Pi_{\max}^{ceil,1}(\theta_0|\theta,\theta_0)=\textit{IR}(\theta_0,1)-\alpha(\theta-\theta_0).\]
Thus,
\begin{align}\nonumber
&\Pi_{\max}^{ceil,1}(\theta|\theta,\theta_0)-\Pi_{\max}^{ceil,1}(\theta_0|\theta,\theta_0)=
\Pi_{\max}^{sym}(\theta|\theta)-\Pi_{\max}^{ceil,1}(\theta_0|\theta,\theta_0)=\\
\nonumber
& \textit{IR}(\theta,1)-\textit{IR}(\theta_0,1)+\alpha(\theta-\theta_0)=
\alpha(\theta-\theta_0)-\textit{IR}(\theta_0,\theta)=\\
\label{ineqq}
& \alpha(\theta-\theta_0)-\alpha\int_{\theta_0}^{\theta} (1-F(u))du
=\alpha\int_{\theta_0}^{\theta}F(u)du>0,
\end{align}
where the first equality follows from \eqref{asym=sym} and the
inequality follows from $\theta>\theta_0$. This means that the deviation from $\tau=\theta$ to $\tau=\theta_0$ is unprofitable as well, so bidding $\tau=\theta$ is optimal, as desired.

The above analysis demonstrates the role of the kickback~--- without it the final inequality would not necessarily hold and some types $\theta>\theta_0$ would switch to bidding $\bar{S}$ with guaranteed victory and choose a different (higher) quality as a result; the desired quality schedule would not be implemented.

\textbf{Case 1b. Favored bidder, true type $\theta<\theta_0$.}
According to the proposed equilibrium, such favored bidder types should find it optimal to bid $\bar{S}=S^*(\theta_0)$, i.e., $\tau=\theta_0$.

For $\theta<\theta_0$, the inequality \eqref{ineqq} is reversed, so
\[\Pi_{\max}^{ceil,1}(\theta_0|\theta,\theta_0)>\Pi_{\max}^{sym}(\theta|\theta).\]
But by \cref{symmimplement}, $\Pi_{\max}^{sym}(\theta|\theta)>\Pi_{\max}^{sym}(\tau|\theta)$ for all $\tau\neq \theta$, including $\tau>\theta_0$. Thus, for $\tau>\theta_0$,
\[\Pi_{\max}^{ceil,1}(\theta_0|\theta,\theta_0)>\Pi_{\max}^{sym}(\theta|\theta)>\Pi_{\max}^{sym}(\tau|\theta)=\Pi_{\max}^{ceil,1}(\tau|\theta,\theta_0),\]
where the last equality follows from \eqref{asym=sym}.
This implies the unprofitability of all deviations $\tau>\theta_0$.

\textbf{Case 2a. Unfavored bidder, true type $\theta>\theta_0$.}
By \cref{symmimplement}, $\Pi_{\max}^{sym}(\theta|\theta)\geq \Pi_{\max}^{sym}(\tau|\theta)$ for all $\tau$, including $\tau\geq \theta_0$. Thus, by \eqref{asym=sym}, 
\[\Pi_{\max}^{ceil,2}(\theta|\theta,\theta_0)\geq \Pi_{\max}^{ceil,2}(\tau|\theta,\theta_0)\]
for all $\tau\geq \theta_0$. Thus, it is optimal for the unfavored bidder with type $\theta>\theta_0$ to bid $S^*(\theta)$ among all $S\leq \bar{S}$, as desired.

\textbf{Case 2b. Unfavored bidder, true type $\theta<\theta_0$.}
By \cref{single-peaked} (single-peakedness of $\Pi_{\max}^{sym}(\tau|\theta)$), for $\theta<\theta_0\leq \tau$ we have
$\Pi_{\max}^{sym}(\theta_0|\theta)\geq \Pi_{\max}^{sym}(\tau|\theta)$. Hence, by \eqref{asym=sym}, 
\[\Pi_{\max}^{ceil,2}(\theta_0|\theta,\theta_0)\geq \Pi_{\max}^{ceil,2}(\tau|\theta,\theta_0)\]
for $\tau\geq \theta_0$. Thus,  it is optimal for the unfavored bidder with type $\theta<\theta_0$ to bid $S^*(\theta_0)$ among all $S\leq \bar{S}$, as desired.

Finally, let us note that in all 4 cases the type $\theta=1$ gets a profit of $\Pi_{\max}^{sym}(1|1)=0$, so the mechanism is individually rational, and the least-efficient type gets zero profit.
\end{proof}

\noindent
\textbf{Derivation for the constant-elasticity example:}
\begin{proof}
It follows from \cref{g'''} that if $\gamma>2$ an optimal mechanism is either score floors, sole-sourcing, or symmetric while if $\gamma<2$ an optimal mechanism is either score ceilings, sole-sourcing, or symmetric.

For a score floors mechanism, the FOC for the optimal threshold $\theta_0$ is 
\begin{equation}\label{focfloor}
\varphi(1-\theta_0,\theta_0)=\int_{\theta_0}^1x^*(1-\theta_0,\theta)d\theta.
\end{equation}
Plugging in $\varphi(z,\theta)=\frac{\gamma-1}{\gamma}z^{\frac{\gamma}{\gamma-1}}-{2\alpha}\theta z$ and  $x^*(z,\theta)=\varphi'_z=z^{\frac{1}{\gamma-1}}-{2\alpha}\theta$, we get an equation which is under $\theta_0<1$ simplified to
\[(1-\theta_0)^{\frac{\gamma-2}{\gamma-1}}=\frac{1}{{\alpha}\gamma}.\]

Sole-sourcing corresponds to $\theta_0=0$ and it obtains if the derivative of the objective at  $\theta_0=0$ is non-positive. Thus, it will obtain if $1-\frac{1}{{\alpha}\gamma}\leq 0$, ${\alpha}\leq 1/\gamma$, otherwise there will be a $\theta_0\in(0,1)$ satisfying FOC and giving a score floors mechanism.

For a score ceiling mechanism, the FOC for the optimal threshold $\theta_0$ is
\begin{equation}\label{focceiling}
\varphi(1,\theta_0)-\varphi(1-\theta_0,\theta_0)=\int_0^{\theta_0}x^*(1-\theta_0,\theta)d\theta.
\end{equation}
Plugging in $\varphi$ and $x^*$, we get
\begin{equation}\label{focleft}
\frac{\gamma-1}{\gamma}\left(1-(1-\theta_0)^{\frac{\gamma}{\gamma-1}}\right)-{\alpha}\theta_0^2=\theta_0(1-\theta_0)^{\frac{1}{\gamma-1}}.  
\end{equation}
In general, there is no closed-form solution to this equation. %
Sole-sourcing now corresponds to $\theta_0=1$ while the symmetric solution to $\theta_0=0$. To obtain an ${\alpha}$ threshold separating the sole-sourcing region from the score ceiling region, plug in $\theta_0=1$ to \eqref{focleft}. To obtain an ${\alpha}$ threshold separating the symmetric solution region from the score ceiling region, solve \eqref{focleft} for ${\alpha}$ and then take the limit when $\theta_0\to0$.
\end{proof}

\section{Proofs for \cref{sec:additional}}
\label{app:additional}
\textbf{Proof of \cref{compstat1}:}
\label{proofofcompstat}
\begin{proof}
Recall that $J(\theta)\equiv \theta+F(\theta)/f(\theta)$. By Assumption~\ref{regularity}, $J(\theta)$ is increasing. Under Assumption~\ref{separablestructure}, $\varphi(z,\theta)=\max_q(q\cdot z-g(q))-\alpha J(\theta)z$. Denote $H(z):=\max_q(q\cdot z-g(q))$. 

Consider the determination of the optimal threshold $\theta_0$ for the score floor mechanism first. Rewriting \eqref{focfloor}, we get that
   \begin{gather*}
      U'_{\theta_0}=\varphi(1-F(\theta_0),\theta_0)f(\theta_0)-f(\theta_0)\int_{\theta_0}^1 \varphi_z(1-F(\theta_0),\theta)d\theta=\\
      f(\theta_0)\left(H(1-F(\theta_0))-\alpha J(\theta_0)(1-F(\theta_0))-\int_{\theta_0}^1 (H'(1-F(\theta_0))-\alpha J(\theta))f(\theta)d\theta\right).
   \end{gather*}
 Thus, \[U''_{\theta_0,\alpha}=f(\theta_0)(1-F(\theta_0))\left(\frac{\int_{\theta_0}^1 J(\theta)f(\theta)d\theta}{1-F(\theta_0)}-J(\theta_0)\right)>0,\]
 where the inequality follows from the fact that $o(\theta)$ is strictly increasing. By the standard monotone comparative statics theorem, the optimal $\theta_0^*(\alpha)$ must be a weakly increasing function.

Now consider the score ceiling mechanism. Rewriting ..., we get that
 \begin{gather*}
U'_{\theta_0}=\varphi(1,\theta_0)f(\theta_0)-\varphi(1-F(\theta_0),\theta_0)f(\theta_0)-
f(\theta_0)\int_0^{\theta_0}\varphi_z(1-F(\theta),\theta)f(\theta)d\theta=\\
      f(\theta_0)\left(H(1)-\alpha J(\theta_0)-H(1-F(\theta_0))+\alpha J(\theta_0)(1-F(\theta_0)-\int_{0}^{\theta_0} (H'(1-F(\theta_0))-\alpha J(\theta))f(\theta)d\theta\right).
   \end{gather*}
 Thus, \[U''_{\theta_0,\alpha}=f(\theta_0)F(\theta_0)\left(\frac{\int_0^{\theta_0}J(\theta)f(\theta)d\theta}{F(\theta_0)}-J(\theta_0)\right)<0,\]
  where the inequality follows from the fact that $o(\theta)$ is strictly increasing. By the standard monotone comparative statics theorem, the optimal $\theta_0^*(\alpha)$ must be a weakly decreasing function.
   \end{proof}
  
\noindent
\textbf{Proof of \cref{effvsopt}:}
\begin{proof}
The efficient mechanism is characterized by a statement analogous to \cref{main result} in which in the definition of the function $\xi(\cdot)$ one replaces $J(\theta)$ with just $\theta$. With $F(\theta)=\theta$ $J(\theta)$ is proportional to $\theta$ and thus \cref{g'''} applies fully to the efficient mechanisms as well. Hence, an optimal mechanism is either a ``score floors'' mechanism or a ``score ceilings'' mechanism. Now note that the efficient mechanism in a given setting 1 is the same as a buyer-optimal mechanism in another setting 2 such that \emph{virtual} costs in setting 2 are equal to costs in setting 1. 
If $C^P={\alpha}\theta$, the \emph{virtual} production costs $C^P+C^P_{\theta}F/f$ are exactly ${2\alpha}\theta$. Thus, the efficient mechanism for $C^P={2\alpha}\theta$ is equal to the buyer-optimal mechanism for $C^P={\alpha}\theta$. So we need to compare the buyer-optimal mechanism for $C^P={\alpha}\theta$ with the buyer-optimal mechanism for $C^P={2\alpha}\theta$. Since $\alpha$ halves and thus decreases, by \cref{compstat1} the optimal threshold moves in a way that makes the mechanism less symmetric. Thus, the efficient mechanism is less symmetric than the buyer-optimal mechanism.
\end{proof}

\noindent
\textbf{Proof of \cref{1orall}}
\begin{proof}
Recall that $\varphi(z,\theta)\equiv \max_x\left(zx-C(x,\theta)\right)$.

For $C(x,\theta)=g(x+{2\alpha}\theta)$, $\varphi(z,\theta)=H(z)-{2\alpha}\theta z$ for some function $H(z)$. Denote $h(e):=H(z)/z$. First, we show that the condition that $\frac{g(t)}{\sqrt{g'(t)}}$ is strictly increasing implies that $h'(z)z^{3/2}$ is increasing.

Indeed, given that $g(0)=0$ it is easy to show that
\[h(e)=\begin{cases}
0, & z<g'(0);\\
g'^{-1}(z)-\frac{g(g'^{-1}(z))}{z}, & z\geq g'(0).
\end{cases}\]
and thus after simplifications
\[h'(z)z^{3/2}=\begin{cases}
0, & z<g'(0);\\
\frac{g(g'^{-1}(z))}{\sqrt{z}}, & z\geq g'(0).
\end{cases}\]
As $g'^{-1}(z)$ is increasing and $\frac{g(t)}{\sqrt{g'(t)}}$ is increasing by assumption, $h'(z)z^{3/2}$ is increasing.

Now, the principal's payoff from playing the optimal symmetric mechanism among $k$ remaining bidders is
\begin{equation}\label{U_k}
 U(k)=\int_0^1k\left(H((1-\theta)^{k-1})-{2\alpha}\theta(1-\theta)^{k-1}\right)d\theta.   
\end{equation}

We shall show that $U(k)$ is quasi-convex when $k$ is treated as a continuous variable\footnote{The integrand in \eqref{U_k} is \emph{quasi-concave} in $k$. However, a sum (integral) of quasi-concave functions can be quasi-convex and not quasi-concave.} This implies the result.

Using the substitution $y=(1-\theta)^{k-1}$, $U(k)$ may be rewritten as
\[U(k)=\frac{k}{k-1}\int_0^1\left(H(y)-{2\alpha}(1-y^{\frac{1}{k-1}})y\right)y^{\frac{2-k}{k-1}}dy=\frac{k}{k-1}\int_0^1\left(h(y)-{2\alpha}(1-y^{\frac{1}{k-1}})\right)y^{\frac{1}{k-1}}dy.\]
Now substitute $\delta:=\frac{1}{k-1}$. Because this is a monotone transformation, it is sufficient to show that $U(k(\delta))$ is quasi-convex in $\delta$. When $k=1$, we set $\delta=+\infty$.
\[U=(1+\delta)\int_0^1\left(h(y)-{2\alpha}(1-y^{\delta})\right)y^{\delta} dy=(1+\delta)\int_0^1h(y)y^{\delta}dy-{2\alpha}(1+\delta)\int_0^1(1-y^{\delta})y^{\delta}dy.\]
We now compute $U'_{\delta}$. Integrating the first term by parts we get
\[(1+\delta)\int_0^1h(y)y^{\delta}dy=\int_0^1h(y)dy^{1+\delta}=h(1)-\int_0^1h'(y)y^{1+\delta}dy.\]
Thus, \[\left((1+\delta)\int_0^1h(y)y^{\delta}dy\right)'_{\delta}=\int_0^1h'(y)y^{1+\delta}\ln(1/y)dy.\]
The second integral can be computed explicitly: \[(1+\delta)\int_0^1(1-y^{\delta})y^{\delta}dy=\frac{\delta}{2\delta+1}.\]
Thus,
\[\left((1+\delta)\int_0^1(1-y^{\delta})y^{\delta}dy\right)'_{\delta}=\frac{1}{(2\delta+1)^2}=\frac{1}{4}\frac{1}{(\delta+\frac{1}{2})^2}=\frac{1}{4}\int_0^1y^{\delta-\frac{1}{2}}\ln(1/y)dy,\]
where the last equality may be verified by integration by parts.

Tacking stock,
\begin{multline*}
 U'_{\delta}=\int_0^1h'(y)y^{1+\delta}\ln(1/y)dy-\frac{{2\alpha}}{4}\int_0^1y^{\delta-\frac{1}{2}}\ln(1/y)dy\\= \left(\frac{\int_0^1h'(y)y^{1+\delta}\ln(1/y)dy}{\int_0^1y^{\delta-\frac{1}{2}}\ln(1/y)dy}-\frac{{2\alpha}}{4}\right)\int_0^1y^{\delta-\frac{1}{2}}\ln(1/y)dy.   
\end{multline*}
Now, we shall show that the expression in parentheses above is increasing in $\delta$. This will imply that $U'_{\delta}(\delta)$ is of increasing sign and hence $U(\delta)$ is quasi-convex.

Rewrite
\[\frac{\int_0^1h'(y)y^{1+\delta}\ln(1/y)dy}{\int_0^1y^{\delta-\frac{1}{2}}\ln(1/y)dy}=\frac{\int_0^1\left[h'(y)y^{\frac{3}{2}}\right]y^{\delta-\frac{1}{2}}\ln(1/y)dy}{\int_0^1y^{\delta-\frac{1}{2}}\ln(1/y)dy}\]
and consider the following family of density functions on $(0,1)$ parametrized by $\delta$:
\[f(y|\delta)=\frac{y^{\delta-\frac{1}{2}}\ln(1/y)}{\int_0^1y^{\delta-\frac{1}{2}}\ln(1/y)dy}.\]
It is not hard to show that for any $\delta_1<\delta_2$ $f(y|\delta_2)$ first-order stochastically dominates $f(y|\delta_1)$.\footnote{For this, note that $f(y|\delta_2)/f(y|\delta_1)$ is increasing in $y$ and thus apart from endpoints the two densities cross only once, with $f(y|\delta_2)$ crossing $f(y|\delta_1)$ from below.} Thus, for any increasing function $b(y)$ 
\[\mathbb{E}_{y\sim f(\cdot|\delta)}(b(y))=\frac{\int_0^1\left[b(y)\right]y^{\delta-\frac{1}{2}}\ln(1/y)dy}{\int_0^1y^{\delta-\frac{1}{2}}\ln(1/y)dy}\]
is increasing in $\delta$. However, we have shown above that the function $h'(y)y^{3/2}$ is increasing indeed.
\end{proof}

\section{Relation to Zhang (2017)}
\label{zhangzhang}
\subsection*{Informal discussion}
In this section, we compare and contrast our model with that of \cite{zhang2017auctions} and \cite{gershkov2021theory}. Since our settings differ in not one but two respects, we introduce an intermediate hypothetical setting, to facilitate the comparison, see table below.

\medskip

    \begin{tabular}{|c|c|c|c|}
    \hline
        Setting & 1 (Zhang 2017) & 2 & 3 (Ours)\\ \hline
        Agent's action is: & Non-contractible & Contractible & Contractible\\ \hline
        Benefits of action accrue to: & Agent & Agent & Principal\\
        \hline
    \end{tabular}
    
\medskip

Settings 1 and 2 are in general not equivalent. Technically, the source of non-equivalence is the different order of \emph{virtualization} and \emph{optimization over action} that happen in the course of solving the mechanism design problem. Namely, in setting 1 one first performs (agent's) optimization over action and then virtializes the obtained maximized utility while in setting 2~--- when an agent's action becomes contractible~--- one would first virtualize the agent's utility for a given action and after that optimize the principal's utility over agent's action.
However, it is not hard to show that in \emph{additively separable} environments (when the agent's benefits of action are additively separable in action and type) the order of optimization and virtualization does not matter, so the settings 1 and 2 are equivalent.

Settings 2 and 3, on the other hand, are \emph{always} equivalent. Under a contractible agent's action the optimal action and allocation schedules do not depend on whether benefits of an agent's action accrue to the agent or the principal as these benefits can be freely moved from the agent to the principal or the other way around using action-dependent transfers.

\subsection*{A formal reduction to \cite{zhang2017auctions}}
Take an instance of our problem under  \cref{separablestructure}. To convert it to an instance of the problem in \cite{zhang2017auctions} (with more general, possibly non-quadratic investment costs), we proceed in 3 steps. 

Before step 1, we rename ``quality'' $q$ to ``action'' $a$ to achieve consistency with the terminology in \cite{zhang2017auctions}. Thus, instead of quality functions  $q_i(\theta_i)$ we will be talking about action profiles $a_i(\theta_i)$.%

\medskip
\noindent
\textbf{Step 1: moving benefits of action from principal to agent}
\begin{lemm}\label{move_action}
If an agent's action is contractible, the optimal allocation and action schedules  $a_i(\theta_i)$, $z_i(\bm\theta)$ in our problem with principal's gross utility $V(a)$, production costs $C^P(a,\theta)$, and some investment costs are the same as in our problem with principal's gross utility of 0, production costs $C^P(a,\theta)-V(a)$ and the same investment costs.
\end{lemm}
\begin{proof}
Note that $\frac{\partial}{\partial\theta} (C^P(a,\theta)-V(a))= \frac{\partial}{\partial\theta}(C^P(a,\theta))$. Thus, 
the virtual production costs in the new problem are also shifted by $V(a)$, as the production costs themselves. Thus, the objective function in the new problem is
\[\mathbb{E}\sum_{i=1}^n\left(0-(\tilde{C}^P(a_i(\theta_i),\theta_i)-V(a_i(\theta_i)))z_i(\bm\theta)-\tilde{C}^I(a_i(\theta_i),\theta_i)\right),\]
while the objective function in the original problem is
\[\mathbb{E}\sum_{i=1}^n\left((V(a_i(\theta_i))-\tilde{C}^P(a_i(\theta_i),\theta_i))z_i(\bm\theta)-\tilde{C}^I(a_i(\theta_i),\theta_i)\right),\]
which is the same thing. Thus, the optimal action and allocation profiles coincide.
\end{proof}

\noindent
\textbf{Step 2: making action non-contractible}
\begin{lemm}\label{invariance}
If production costs are additively separable, i.e., $C^P(a,\theta)=l(\theta)-b(a)$, the principal's gross utility is 0 and the investment costs do not depend on agent's type, then the optimal allocation $z^*_i(\bm\theta)$ when agents' actions are contractible is the same as when agents' actions are non-contractible, as in \cite{zhang2017auctions}. The optimal action profiles  $a^*_i(\theta_i)$ when agents' actions are contractible are the same as equilibrium action profiles under $z^*_i(\bm\theta)$ when agents' actions are non-contractible.
\end{lemm}

\begin{proof}
With $C^P(a,\theta)=l(\theta)-b(a)$, the virtual production costs are given by $\tilde{C}^P(a,\theta)=l(\theta)-b(a)+l'(\theta)\frac{F(\theta)}{f(\theta)}$. As $C^I_{\theta}=0$, the virtual investment costs equal investment costs themselves, i.e., $\tilde{C}^I(a,\theta)=C^I(a,\theta)$. We can write just $C^I(a)$ for the both virtual and non-virtual investment costs. Thus, when agents' actions are contractible, the principal's objective function is
\begin{equation}\label{P_obj1}
\mathbb{E}\sum_{i=1}^n\left(\left(0-l(\theta_i)-l'(\theta_i)\frac{F(\theta_i)}{f(\theta_i)}+b(a_i(\theta_i))\right)z_i(\bm\theta)-C^I(a_i(\theta_i))\right).  
\end{equation}
The principal maximizes this with respect to $z_i(\bm\theta)$ and $a_i(\theta_i)$. Denote by $\bar{z}_i(\theta_i)$ the interim allocation of agent $i$, that is, $\bar{z}_i(\theta_i)\equiv \mathbb{E}z_i(\bm\theta)$. Suppose the principal first chooses $z_i(\bm\theta)$ and then chooses $a_i(\theta_i)$ given $z_i(\bm\theta)$. Denote by $a_i^P(\theta_i)[z_i(\bm\theta)]$ the  agent's $i$ action which is optimal for the principal given $z_i(\bm\theta)$. From \eqref{P_obj1}, it is clear that given $z_i(\bm\theta)$, the principal will choose
\[a_i^P(\theta_i)[z_i(\bm\theta)]=\arg\max_a\left(b(a)\bar{z}_i(\theta_i)-C^I(a)\right).\]
Now suppose the agents' actions are not contractible. In this case, the principal only chooses $z_i(\bm\theta)$ in \eqref{P_obj1} while $a_i(\theta_i)$ is determined from agent's optimization given $z_i(\bm\theta)$. Denote by $\bar{t}_i(\theta_i)$ the interim payment to the agent in a direct mechanism. Note that it is not conditioned on $a_i$ since the action is not contractible. Given $\bar{z}_i(\theta_i)$ and $\bar{t}_i(\theta_i)$, the agent $i$ of type $\theta_i$ reports the type truthfully and chooses $a_i$ to maximize
\[
\Pi_i= \bar{t}_i(\theta_i) -C^P(a_i,\theta_i)\bar{z}_i(\theta_i)-C^I(a_i)=
\bar{t}_i(\theta_i)-(l(\theta_i)-b(a_i))\bar{z}_i(\theta_i)-C^I(a_i).\]
The agent will clearly choose 
\[a_i^A(\theta_i)[z_i(\bm\theta)]=\arg\max_a\left(b(a)\bar{z}_i(\theta_i)-C^I(a)\right).\]
Thus, $a_i^A(\theta_i)[z_i(\bm\theta)]=a_i^P(\theta_i)[z_i(\bm\theta)]$ for any allocation $z_i(\bm\theta)$. This implies that the principal's objective as a function of $z_i(\bm\theta)$ \emph{only} is the same regardless of whether the agents' actions are contractible or not. Thus, the optimal allocation $z^*_i(\bm\theta)$ is the same regardless of whether the agents' actions are contractible or not. As $a_i^A(\theta_i)[z_i(\bm\theta)]=a_i^P(\theta_i)[z_i(\bm\theta)]$ for any allocation $z_i(\bm\theta)$, the equality holds for the optimal $z^*_i(\bm\theta)$ as well, establishing the second claim of the lemma.
\end{proof}
Note that the above argument does \emph{not} work when the production costs are not additively separable and/or the investment costs depend on type. In this case, the principal's preferred action given $z_i(\bm\theta)$ satisfies
\[a_i^P(\theta_i)[z_i(\bm\theta)]=\arg\max_a\left(-\tilde{C}^P(a,\theta_i)\bar{z}_i(\theta_i)-\tilde{C}^I(a,\theta_i)\right).\]
while under non-contractible actions the agent would choose
\[a_i^A(\theta_i)[z_i(\bm\theta)]=\arg\max_a\left(-C^P(a,\theta_i)\bar{z}_i(\theta_i)-C^I(a,\theta_i)\right),\]
which is in general different from $a_i^P(\theta_i)[z_i(\bm\theta)]$.

\medskip
\noindent
\textbf{Step 3: converting the optimal-procurement problem to an optimal-selling-procedure problem and wrapping up.}
\medskip

Recall \cref{separablestructure}: the principal's utility is $V(a)$, production costs are $C^P(a,\theta)=\alpha\theta$, investment costs are $C^I(a,\theta)=g(a)$ and the distribution of types is $F(\theta)$. To formally relate our model to  \cite{zhang2017auctions}, we need to translate our problem into an optimal-selling problem and also change the definition of type since in \cite{zhang2017auctions}, an agent's utility increases in type while in our problem it decreases in type. To avoid confusion, we denote the type from \cite{zhang2017auctions} by $\mu$ and the distribution of $\mu$ by $F^{\mu}(\mu)$.

\begin{lemm}\label{to_Zhang}
Let $z^*_i(\bm\theta)$ and $a^*_i(\theta_i)$ be the optimal allocation  and action profiles in our problem under  \cref{separablestructure}. Let $\hat{z}_i(\bm\mu)$ be the optimal allocation and $\hat{a}_i(\mu_i)$ be equilibrium action profiles in the problem in \cite{zhang2017auctions} in which an agent's utility is $V(a)+\alpha\mu-\alpha$, investment costs are $g(a)$ and the distribution of types is $F^{\mu}(\mu)=1-F(1-\mu)$.  Then,
$z^*_i(\bm\theta)=\hat{z}_i(\bm{1-\theta})$ and
$a^*_i(\theta_i)=\hat{a}_i(1-\theta_i)$.
\end{lemm}

\begin{proof}
By \cref{move_action}, the optimal allocation and action profiles in our problem are same as in the problem with principal's gross utility of 0, production costs of $\alpha\theta-V(a)$, same investment costs of $g(a)$ and same distribution of types $F(\theta)$. 
Since the cost function $C^P(a,\theta)=\alpha\theta-V(a)$ is additively separable and other conditions of \cref{invariance} hold, by \cref{invariance} the optimal profiles are the same as in the procurement version of the problem in \cite{zhang2017auctions} with these costs functions and distribution of types $F(\theta)$.

To translate this to an original optimal selling problem from \cite{zhang2017auctions}, we redefine type $\mu:=1-\theta$ to make the agent's utility depend positively on type. The cdf of $\mu$ is $F^{\mu}(\mu)=1-F(1-\mu)$. The production costs become $\alpha(1-\mu)-V(a)$. Finally, an agent's utility from consuming a good is naturally the negative production costs.
\end{proof}

\end{document}